\documentclass[journal]{IEEEtran}
\ifCLASSINFOpdf
\else
\fi

\usepackage{algorithm, algorithmic}
\usepackage{graphicx} 
\usepackage{amsmath} 
\usepackage{amssymb}  
\usepackage{amsthm}
\usepackage{amsfonts}
\usepackage{dsfont}
\usepackage{color}
\usepackage{bm}

\usepackage{times,amsmath,epsfig}
\usepackage{cite}
\usepackage{tikz}
\usepackage{xcolor}
\usepackage[english]{babel}
\usepackage[latin1]{inputenc}
\usepackage{subfigure}
\usepackage{hhline}
\usepackage{mathtools}
\usepackage{multirow}

\newcommand{\Phib}{\mathbf{\Phi}}
\newcommand{\psib}{\bm{\psi}}
\newcommand{\phib}{\bm{\varphi}}
\newcommand{\Psib}{\mathbf{\Psi}}
\newcommand{\Ib}{\mathbf{I}}
\newcommand{\Xib}{\mathbf{\Xi}}

\newcommand{\Ab}{\mathbf{A}}
\newcommand{\Bb}{\mathbf{B}}
\newcommand{\Lb}{\mathbf{L}}
\newcommand{\Db}{\mathbf{D}}
\newcommand{\Abb}{\mathbf{T}_{\xb}}
\newcommand{\Gb}{\mathbf{G}}
\newcommand{\Ub}{\mathbf{U}}
\newcommand{\ub}{\mathbf{u}}
\newcommand{\yb}{\mathbf{y}}
\newcommand{\xb}{\mathbf{x}}
\newcommand{\nb}{\mathbf{n}}

\newcommand{\V}{\mathcal{V}}
\newcommand{\W}{\mathcal{W}}
\newcommand{\E}{\mathcal{E}}
\newcommand{\G}{\mathcal{G}}
\newcommand{\C}{\mathcal{C}}
\newcommand{\D}{\mathcal{D}}
\newcommand{\Nc}{\mathcal{N}}
\newcommand{\Ac}{\mathcal{T}}
\newcommand{\Ncb}{\bar{\mathcal{N}}}
\newcommand{\R}{\mathcal{R}}
\newcommand{\Sx}{\mathcal{S}}
\newcommand{\Acc}{\mathcal{\Ac_{\xb}}}
\newcommand{\I}{\mathcal{I}}

\newcommand{\RR}{\mathbb{R}}

\newcommand{\Ed}{\mathds{E}}
\newcommand{\CC}{\mathbb{C}}
\newcommand{\PP}{\mathbb{P}}
\newcommand{\gmin}{g_{\mathrm{\tiny min}}}
\newcommand{\e}{\bm{\varepsilon}}
\newcommand{\gammas}{\gamma_2(\Ac,\lVert \cdot\rVert_2)}
\newcommand{\gammasb}{\gamma_2(\Acc,\lVert \cdot\rVert_2)}
\newcommand{\xc}{\hat{\mathbf{x}}}
\newcommand{\pii}{\pi(i)}
\newcommand{\pij}{\pi(j)}

\newcommand{\vertiii}[1]{{\left\vert\kern-0.25ex\left\vert\kern-0.25ex\left\vert #1 
    \right\vert\kern-0.25ex\right\vert\kern-0.25ex\right\vert}}

\newtheorem{lemma}{Lemma}
\newtheorem{definition}{Definition}

\newtheorem{theorem}{Theorem}

\begin{document}
%
\title{Sampling of graph signals via randomized local aggregations}
%
%
%

\author{Diego~Valsesia,~\IEEEmembership{Member,~IEEE,}
        Giulia~Fracastoro,~\IEEEmembership{Member,~IEEE,}
        and~Enrico~Magli,~\IEEEmembership{Fellow,~IEEE}
\thanks{The authors are with Politecnico di Torino -- Department of Electronics and Telecommunications, Italy. Email: \{name.surname\}@polito.it.}

}

\maketitle

\begin{abstract}
Sampling of signals defined over the nodes of a graph is one of the crucial problems in graph signal processing. While in classical signal processing sampling is a well defined operation, when we consider a graph signal many new challenges arise and defining an efficient sampling strategy is not straightforward. Recently, several works have addressed this problem. The most common techniques select a subset of nodes to reconstruct the entire signal. However, such methods often require the knowledge of the signal support and the computation of the sparsity basis before sampling. Instead, in this paper we propose a new approach to this issue. We introduce a novel technique that combines localized sampling with compressed sensing. We first choose a subset of nodes and then, for each node of the subset, we compute random linear combinations of signal coefficients localized at the node itself and its neighborhood. The proposed method provides theoretical guarantees in terms of reconstruction and stability to noise for any graph and any orthonormal basis, even when the support is not known. Code is available at \verb!https://git.io/fj0Ib! 
\end{abstract}


\section{Introduction}
\label{sec:intro}

The analysis of signals defined over graphs is a natural extension of classical processing of time-varying signals that has emerged due to the wide array of applications where the signals of interest lie on irregular domains \cite{shuman2013emerging}. Fields as varied as sensor and social networks \cite{sensor_networks,Lazer721}, genetics \cite{Davidson1669}, neuroscience \cite{huang2016graph}, computer graphics \cite{thanou2016graph} and image processing \cite{fracastoro2015predictive,fracastoro2016graph,hu2015multiresolution} may benefit from graph signal processing (GSP) techniques. The pervasiveness of big data \cite{Moura_bigdata} makes sampling and subsequent recovery crucial to creating efficient representations. While sampling is well understood in classical signal processing for time-varying signals, not least thanks to the recent developments in compressed sensing (CS) \cite{Donoho_CS}, signals whose coefficients lie on the nodes of a graph pose a unique set of challenges. Most current efforts in the field have been focused on establishing a concept of a frequency domain for graph signals as a linear transform dependent on the graph structure \cite{hammond2011wavelets, Leuss_Dual}, and uncovering the fundamental limits of sampling for signals having a low-dimensional structure in this transformed domain \cite{Tsitsvero_uncertainty}. Efficient sampling of graph signals is concerned with a notion of locality of the sampling procedure where only a few nodes can be directly observed and such nodes can only return the value of the coefficient on the node itself \cite{Puy_random,Chen_sampling,anis2016efficient,romero2017kernel,chamon2017greedy} or an aggregate value determined from the coefficients on its neighbors \cite{Marques_successive,wang2016local}. Universal sampling strategies such as CS with dense random projections would be optimal for graph signals as well as classical signals because they are independent of the particular transform enforcing the low-dimensional prior \cite{Candes_nearoptimal}. However, dense random projections violate locality as each projection requires to combine all the coefficients across the graph. This paper shows that it is possible to combine locality with CS to sample graph signals with a low-complexity scheme that provides theoretical guarantees for perfect reconstruction and stability to noise. In particular, the proposed method of randomized local aggregations fits the line of works on aggregation sampling and consists of computing random linear combinations of the signal coefficients restricted to the neighborhood of a subset of nodes in the graph. The theoretical analysis rigorously derives conditions on the number of measurements required to guarantee perfect reconstruction and stability to noise depending on the structure of the graph and of the signal. Finally, the proposed sampling method does not need a priori knowledge of the sparsity support of the signal and can be therefore used even if it is unknown at sampling time. This is common to many problems, as discussed in the literature on CS \cite{Donoho_CS}, as well as on GSP \cite{Marques_successive,Varma_spectrum}.

This paper is organized as follows. Section \ref{sec:bkg} presents some background on GSP and CS, and discusses relevant works in the literature. Section \ref{sec:method} introduces the proposed method performing randomized local aggregations. Section \ref{sec:theory} analyzes the method and provides theoretical guarantees. Section \ref{sec:experiments} reports numerical experiments comparing the proposed method with relevant methods in the literature. Finally, Section \ref{sec:conclusions} draws some conclusions.

\subsection*{Notation}

Lower-case (upper-case) bold symbols denote real-valued vectors (matrices). Lower-case letters indicate scalars. Upper-case letters denote random variables. Symbols $\PP$ and $\Ed$ denote the probability and expectation operators, respectively. Calligraphic letters denote sets. Matrix $\Ab_\R$ indicates the restriction of $\Ab$ to rows indexed by set $\R$, while $\Ab_{\vert\C}$ indicates restriction of $\Ab$ to columns indexed by set $\C$.

\section{Background}
\label{sec:bkg}

\subsection{Graph signal processing}
\label{sec:bkg_graph}
We first introduce a few definitions related to graph signal processing that will be used in the remainder of the paper.
Let $\G(\V,\E)$ denote a directed graph where $\V$ is the set of nodes with $|\V|=n$ and $\E$ is the set of edges, being $(i,j)$ an edge from node $i$ to node $j$. 
\begin{definition}
Given $\G$, for each $i\in\V$ we define the in-neighborhood of the node $i$ as $$\Nc_i=\{j\in\V : (j,i)\in\E\}.$$ 
\end{definition}
Throughout the paper, whenever the terms neighborhood or neighbor are used, they are meant to be an in-neighborhood or an in-neighbor, respectively. For an undirected graph, the in- and out- neighborhoods coincide.
\begin{definition}
Given $\G$, a dominating set of the graph $\G$ is a subset $\D\subseteq\V$ with the following property
\[
\forall i\in\V\ \exists\,j\in\D\ \mathrm{ s. t. }\ i\in\Ncb_j,
\]
where $\Ncb_j=\Nc_j\cup\{j\}$.
\end{definition}
Given $\G$, we also introduce the adjacency matrix $\Ab\in\{0,1\}^{n\times n}$ where $a_{ij}=1$ if $(i,j) \in \E$ and $a_{ij}=0$ otherwise. We define $\bar{\Ab}$ as the adjacency matrix of the same graph $\G$ including self-loops, i.e. $\bar{a}_{ii}=1$ $\forall{i}\in\V$. 
For undirected graphs, the graph Laplacian is defined as $\Lb=\Db-\Ab$ where $\Db$ is a diagonal matrix whose $i$-th diagonal element $\Db_{ii}$ is the sum of all the edges incident to node $i$. 
$\Lb$ is diagonalizable by an orthogonal matrix $$\Lb=\Ub\mathbf{\Lambda}\Ub^H,$$ where $\Ub\in\CC^{n\times n}$ is the eigenvector matrix of $\Lb$ that contains the eigenvectors as columns and $\mathbf{\Lambda}\in\RR^{n\times n}$ is the diagonal eigenvalue matrix where the eigenvalues are sorted in increasing order. 
A graph signal $\xb\in\RR^{n}$ in the node domain is a real-valued function defined on the nodes of the graph $\G$, such that $x_i$ with $i=1,...,n$ is the value of the signal at node $i\in\V$ \cite{shuman2013emerging}. For an undirected graph, the eigenvectors of the Laplacian are used to define the graph Fourier transform (GFT) \cite{shuman2013emerging} of the signal $\xb$ as follows 
\[
\hat{\xb}=\Ub^H\xb.
\]
If we consider a directed graph, we can define the GFT matrix using the eigenvector matrix of the adjacency matrix $\Ab$, in this case we have $\Ub\in\CC^{n\times n}$ \cite{sandryhaila2014discrete}.

\subsection{Compressed sensing}
\label{sec:bkg_cs}

Compressed sensing (CS) \cite{Donoho_CS, Candes_nearoptimal} is an established theory for signal sampling relying on linear sampling operators satisfying certain norm-preservation properties and employing nonlinear reconstruction techniques that leverage signal models such as sparsity. The typical setup involves a signal $\xb \in \RR^n$, having a sparse representation under an orthonormal basis $\Ub$, $\xb = \Ub \hat{\xb}$, with $\Vert \hat{\xb} \Vert_0 \leq k$, and a vector of measurements $\yb \in \RR^m$, $m \ll n$, computed as 
\begin{align*}
\yb = \Phib \xb = \Phib \Ub \hat{\xb} = \Psib \hat{\xb}.
\end{align*}
The most well-known reconstruction algorithm to recover $\hat{\xb}$ from $\yb$ is $\ell_1$ norm minimization:
\begin{align}
\label{eq:l1min}
\hat{\xb}^* = \arg\min_{\hat{\xb}} \Vert \hat{\xb} \Vert_1 \quad \text{s.t.} \quad \Psib \hat{\xb} = \yb,
\end{align}
with $\Psib = \Phib \Ub$.

It is known that a sufficient condition for perfect reconstruction and stability in presence of additive noise is that matrix $\Psib$, the product between the sampling matrix and the sparsity basis, provides a stable, distance-preserving embedding of the set of sparse signals $\Sigma_k$. This notion is captured by the Restricted Isometry Property (RIP) \cite{Candes_decoding, Baraniuk2008, Mendelson2008}.
\begin{definition}[RIP] A matrix $\Psib$ satisfies the Restricted Isometry Property of order $k$ if there exists a $\delta_k \in (0,1)$ such that
\begin{align*}
(1-\delta_k)\lVert \hat{\xb}\rVert_2^2\le\lVert\Psib\hat{\xb}\rVert_2^2\le(1+\delta_k)\lVert \hat{\xb}\rVert_2^2
\end{align*}
holds for all $\hat{\xb} \in \Sigma_k$
\end{definition}

Several random constructions for the sampling matrix $\Phib$ have been studied. Each specific construction interacts in a different way with the sparsity basis $\Ub$. More structured constructions typically need a larger number of measurements for $\Psib$ to satisfy the RIP with constant probability. It is known that a sampling matrix $\Phib$ made of independent and identically distributed (i.i.d.) sub-Gaussian random variables is universal, meaning that the RIP holds with $m=O(k \log n)$ regardless of the basis $\Ub$. However, dense random matrices are not practical in many scenarios due to excessive memory requirements or architectural constraints in the sampling procedure (e.g., in the graph setting, they would violate locality). Toeplitz \cite{Haupt_Toeplitz_IT}, circulant \cite{JLCirculant_improved} and block-diagonal \cite{Yap_block,eftekhari2015restricted} matrices are more structured random matrices suitable for common acquisition systems \cite{Duarte_structured, Valsesia_compressive} and they are also known to satisfy the RIP, albeit requiring more measurements and not being uniformly good for all bases.  

\subsection{Related work}
\label{sec:bkg_related}

Several works have addressed the problem of sampling signals defined on graphs, under the assumption of sparsity in a transform domain, i.e., only up to $k$ nonzero coefficients in the positions defined by the support. The most common model assumes a bandlimited signal, i.e., a signal where the sparsity basis is the GFT and where the support coincides with the first $k$ frequency coefficients. 
The sampling problem in the graph signal setting is typically concerned with trading off locality and reconstruction accuracy. Locality implies that the signal is observed only at a small number of nodes or through short-range interactions of a node with its neighbors. We will use the general term \emph{measurement} to refer to either a pointwise selection of a signal coefficient or the result of a more complex operation such as a linear combination. The result of the sampling operation is a vector of measurements such that the original signal can be reconstructed via linear or nonlinear interpolation methods. 

Puy et al. \cite{Puy_random} show that random selection of nodes can be an effective technique. Indeed, Chen et al. \cite{Chen_sampling} proved that uniform random sampling is optimal for the Erd\H{o}s-R\'{e}nyi graph, thus requiring only $k$ measurements. However, uniform random sampling is in general suboptimal because the required number of measurements is heavily affected by the structure of the graph. Puy et al. \cite{Puy_random} improve over uniform sampling by proposing a weighted sampling scheme that uses information on the signal support and the transform basis to bias the sampling procedure towards nodes where the signal is highly localized. Albeit improving over uniform sampling, this method requires knowledge of the support and computation of the sparsity basis before sampling. It is also suboptimal as it generally requires more than $k$ measurements.

Tsivero et al. \cite{Tsitsvero_uncertainty} study node-frequency uncertainty principles to derive necessary conditions for signal recovery after pointwise sampling. They also discuss greedy node selection strategies that approximate the optimal sampling objective, aimed at minimizing the MSE of reconstruction. Such greedy strategies require the knowledge of the signal support and are typically computationally expensive, e.g., requiring multiple singular value decompositions.

Sakiyama et al. \cite{sakiyama_2016,sakiyama_2017} study node selection methods in the context of sensor position selection. In \cite{sakiyama_2016} an assumption of a bandlimited signal is made and nodes are selected so that the cutoff frequency of the signal restricted to the sampling set is maximized. This is done using three possible techniques requiring access to the eigenvectors of the graph Laplacian. Thus this method only targets bandlimited signals (a special case of known support) and requires knowledge of the basis at sampling time. In \cite{sakiyama_2017} a different method is proposed where the localization operator is used to choose nodes so that the operators localized at the nodes of choice cover well the whole graph with minimal overlap. This is a rather efficient method since the localization operator can be implemented with polynomial filters, thus not requiring the full basis. While the method does not strictly need support information and can in principle be applied even if the signal support is unknown, there is an implicit assumption on the support in the choice of the kernel for the localization operator. In fact, the exponential kernel used by the authors implicitly assumes a bandlimited model, and it is not obvious how to choose a kernel for an arbitrary support.

Marques et al. \cite{Marques_successive}, motivated by an analogy with classical sampling, propose a sampling strategy computing measurements as the result of iterated applications of the graph shift operator (e.g., the adjacency matrix) observed at one or more nodes. This strategy can guarantee perfect reconstruction with the optimal number of measurements and does not require knowledge of the support of the signal, so it can also be used when the support is unknown. The technique is equivalent to linear sampling where the sampling matrix has a Vandermonde structure. Unfortunately, Vandermonde matrices are known to be poorly conditioned \cite{Pan_vandermonde} resulting in severe instability even for moderately-sized graphs. The analogy drawn by the authors with classical sampling does not consider that the Vandermonde structure arising in classical sampling is a unique special case that is well conditioned, as it results in the DFT \cite{Pan_vandermonde}. Moreover, the method is highly dependent on the specific nodes chosen as observation points and it is not clear how to perform the selection.

In this paper, we present a sampling method based on random linear combinations of signal coefficients localized at a subset of nodes and their neighbors. The method is not computationally expensive (e.g., it does not require the sparsity basis to be known) and non-adaptive, i.e., it works without knowledge of the signal support. We provide a sufficient condition to select the sampling set of nodes, which must be a superset of a dominating set of the graph. Thanks to such sampling nodes and the use of random coefficients in the linear combinations, we are able to prove the RIP for the resulting sampling matrix. The proposed method can achieve the optimal number of measurements when the support is known. When the support is unknown, the RIP guarantees stability provided a sufficient number of measurements is computed. Such number depends on the structure of the graph and the sparsity basis, as shown by the theoretical analysis in Section \ref{sec:theory}. Compared to the previously discussed approaches in the literature, the main advantage of the proposed method is its generality, i.e., the ability to work both when the signal support is known and when it is unknown, having at the same time theoretical guarantees for perfect and stable reconstruction for any kind of signal support. Such wide-reaching guarantees on performance are only possible since this is a non-adaptive aggregation sampling scheme. Schemes based on selection sampling of nodes may not require cooperation among nodes to exchange information but are adaptive, in the sense that they require more stringent conditions on the signals such as prior knowledge of the signal support or sparsity basis at sampling time or more specific signal models beyond sparsity, such as bandlimitedness. In a nutshell, techniques based on aggregation sampling, and especially the randomized scheme proposed in this paper, trade communications for generality.

\section{Randomized local aggregations}
\label{sec:method}

This section describes the proposed technique to sample graph signals via random linear combinations of signal coefficients localized to neighborhoods of a subset of the nodes of the graph, which we refer to as random local aggregations. The main idea is to select a subset of nodes and at each of them aggregate the signal coefficients of the node and its neighbors with a random linear combination in order to obtain a measurement. The node selection procedure and the use of random linear combinations allows a sampling procedure that is simultaneously \emph{localized}, i.e., it only requires interactions of the chosen node with its neighbors, and can guarantee perfect reconstruction and stability to noise. In particular, the proposed technique is equivalent to randomized linear sampling, similar to the random projections used in CS, with a sampling operator that satisfies the RIP, albeit not being dense, in general.

The proposed method does not rely on the knowledge of the sparsity support of the signal in a transformed domain (e.g., the knowledge of the active set of frequencies in the GFT) and can be used both when that is known or unknown. We want to compute a number of measurements $m$ greater than or equal to the sparsity $k$. When the support is known, choosing $m\geq k$ improves robustness to noise. Otherwise, this condition is needed to ensure correct recovery (the optimal $\ell_0$ minimization requires at least $2k$ measurements). The theoretical analysis in section \ref{sec:theory} shows that the proposed technique can obtain perfect reconstruction in absence of noise and with known support with the optimal number of measurements (i.e., $m=k$), for any graph signal with a dense sparsity basis, contrary to other techniques requiring a larger number of measurements unless the graph has a particular structure (e.g., Erd\H{o}s-R\'{e}nyi).

As a practical application example, let us think of the problem of deploying a sensor network having many low-power sensors that are equipped with short-range transmitters and a small number of (more expensive) sensors with long-range transmitters. The sensors acquire a signal defined over the nodes of a graph that may represent their communication links, and the goal is to compute a compact representation of the signal to be transmitted to a remote fusion center by the long range transmitters. The proposed algorithm allows that and requires a smaller number of long-range transmitters (or equivalently, a smaller data payload for the same number of transmitters). At the same time, no prior knowledge on the signal to be sampled is needed in the design phase, thus being suitable even for signals with non-stationary supports. 

\subsection{Measurements}
We consider a signal $\xb \in \RR^n$, defined over a directed graph $\G(\V,\E)$ with $|\V|=n$ and $\bar{\Nc}_i$ the in-neighborhood of node $i$ including a self-loop (from here onwards, the term neighborhood always refers to the in-neighborhood). The signal has a sparse representation under an orthonormal basis $\Ub$, $\xb = \Ub \hat{\xb}$ with $\Vert \hat{\xb} \Vert_0 \leq k$ and support $\Sx = \lbrace i : \hat{x}_i \neq 0 \rbrace$. We highlight that the method is highly general and we can choose any orthonormal basis $\Ub$. One example is the GFT matrix defined in Section \ref{sec:bkg_graph}. Let $m$ be the desired number of measurements and $\R \subseteq \V$ be the sampling set, i.e., the set of nodes at which measurements are computed. We also define a function $\pi: \R \to [1,m]$ mapping a node $i \in \R$ to its corresponding measurement index $\pi(i)$. Then, a measurement is computed for each $i \in \R$ as:
\begin{align*}
y_{\pi(i)} = \sum_{j \in \bar{\Nc}_i} \varphi_{\pi(i)j} x_j ,
\end{align*}
being each $\varphi_{\pii j}$ a zero-mean independent Gaussian random variable drawn as
\begin{align*}
\varphi_{\pii j} \sim N \left( 0,\frac{1}{g_j} \right)
\end{align*}
with 
\begin{align}
\label{eq:g}
g_j= \sum_{i \in \R} \vert \bar{\Nc}_{i} \cap \lbrace j \rbrace \vert
\end{align}
being the node multiplicity, i.e., the number of times node $j$ appears in the linear combinations. Notice that the node multiplicity only depends on the graph topology and the sampling set, which are known and precomputed before observing the signal.

In matrix form, the measurements are computed as
\begin{align*}
\yb = \Phib \xb
\end{align*}
where the sampling matrix $\Phib$ is defined as follows
\begin{equation}
\label{eq:def_phi}
\varphi_{\pii j}\sim\begin{cases} N \left( 0,\frac{1}{g_j} \right),& \mbox{if }j\in\Ncb_i\\
\delta_0,& \mbox{if }j\notin\Ncb_i\end{cases},
\end{equation}
where $\delta_0$ is a Dirac's delta distribution centered in 0. For an undirected graph $\Phib$ can be seen as the elementwise product of the adjacency matrix of the graph augmented to include self-loops and restricted to the rows of the sampling set $\bar{\Ab}_\R$ and a random matrix $\Xib$ with independent entries, $\Phib = \bar{\Ab}_\R \odot \Xib$.

\subsection{Selection of the sampling set}
\label{sec:samp_strat}
The sampling set must be chosen to ensure that the the norm of the signal is preserved in expectation, i.e. $\Ed \left[ \Phib^H\Phib \right] = \Ib$. A sufficient condition to ensure this is that the sampling set $\R$ is a superset of a dominating set $\D$ of the graph, i.e., $\R \supseteq \D$. 
Choosing the sampling set $\R$ to be a dominating set ensures that the sampling matrix $\Phib$ has no empty columns, since all the nodes (including those in $\R$, thanks to self-loops) have at least a neighbor in $\R$. Moreover, all the rows have random coefficients from a continuous distribution to ensure that they are linearly independent with probability 1, as proved in Theorem \ref{thm:rank}.

Given a graph $\G$, finding a dominating set of minimal size is known to be NP-hard \cite{Garey_computers}. However, good approximations can be found efficiently with greedy algorithms. For instance, an approximation within $O(\log \Delta)$ of the minimal size can be found with complexity $O(\log\vert \V\vert \log \Delta)$ being $\Delta$ the maximum vertex degree \cite{Jia2002}.

Depending on the desired number of measurements $m$ and the size of the dominating set found by the particular algorithm in use, three situations may occur:
\begin{enumerate}
\item $\vert \D \vert = m$: in this case $\R = \D$;
\item $\vert \D \vert > m$: we use $p$-hop random aggregations, i.e., define an equivalent graph having as adjacency matrix $\Bb$ the union of all $\Ab^l$ powers with $l = 1,\dots,p$, i.e., $b_{ij} = 1$ if and only if $\exists\  l\in [1,p]$ s.t. $a_{ij}^l=1$; $p$ is the minimum positive integer such that $\D^p$, the dominating set computed from $\Bb$, has cardinality $\vert \D^p \vert \leq m$.
\item $\vert \D^p \vert < m$: we add rows to $\Phib$ according to one of two strategies
\begin{itemize}
\item repetition of dominating nodes; 
\item insertion of new nodes.
\end{itemize}
\end{enumerate}
The two strategies essentially differ by which nodes are considered eligible for insertion to increase the number of measurements. Let us call $\I$ the set of nodes eligible for insertion.

In the former strategy, a dominating node is repeated\footnote{With slight abuse of notation we suppose that $\R$ contains two or more copies of the same node.}, i.e., $\I = \D^p$. If node $i \in \D^p$ is to be repeated, a new random combination of the signal coefficients in the neighborhood $\bar{\Nc}_i$ is computed. However, drawing new random coefficients does not guarantee that the new measurement is linearly independent as there may be a subgraph for which we are computing more measurements than signal coefficients. To avoid this, a solution is to have a greedy algorithm that before inserting the node in $\R$ checks if the rank of the resulting sampling matrix is full or not. This strategy is advantageous in terms of localization since it does not require increasing the number of transmissions within the nodes of graph, nor does it require observing more than $\vert \D^p \vert$ nodes. Notice, however, that checking if the new measurement is linearly independent can be expensive (up to $O(n^3)$ with Gaussian elimination) and it is not practical for large graphs. 

On the other hand, the latter strategy considers eligible for insertion only nodes that are not already in the sampling set, i.e., $\I = \V \setminus \R$. Insertion of new nodes always guarantees that the new measurement is linearly independent (the limit case being cliques forming square submatrices which however are still full rank with probability 1 due to the independent random coefficients). The complexity of this strategy is therefore limited to solving \eqref{eq:node_insert}, which is $O(n)$. However, this strategy reduces locality, up to $m=n$ when all nodes are in the sampling set.

The choice of the order to use for node insertion under both strategies is determined by the theoretical analysis in Section \ref{sec:theory}. In particular, we choose node $i^*$ as the eligible node having most neighbors with lowest node multiplicity $g$: 
\begin{align}
\label{eq:node_insert}
i^* &= \arg\max_{i \in \I}\ \Bigl| \Bigl\lbrace j : j \in \bar{\Nc}_i , g_j = g_\text{min}^\I \Bigr\rbrace \Bigr|,\\g_\text{min}^\I &= \min_{\ell \in \bigcup_{i \in \I} \bar{\Nc}_i } g_\ell. \nonumber
\end{align}
Informally, this is the node which has most neighbors that have been least seen by the various local aggregations, thus ensuring that all the signal coefficients have been observed enough times. This is particularly important when the signal is highly localized to a set of nodes that are not very well connected with the rest of the graph. Notice that the proposed criterion to select the node bears some similarity with the entropy-based criterion in \cite{Krause2008}. However, it differs in the signal model adopted, being sparse in a transform basis rather than a Gaussian process.

Notice that we have proposed $p$-hop random aggregations when $\vert \D \vert > m$. They represent a tradeoff between number of measurements and locality. For instance, in a sensor network the sampling nodes might be the only ones equipped for long-range communication.   
Under this example, since the proposed method computes linear combinations of coefficients in some neighborhood, such coefficients must be transmitted from the nodes to the nodes in the sampling set. This is also common to other methods performing aggregations, such as \cite{Marques_successive, wang2016local}. Defining one transmission as a real-valued scalar crossing an edge in the graph, the total number of transmissions for the proposed method depends on the strategy that is used to achieve $\vert \R \vert = m$:
\begin{itemize}
\item Repetition of dominating nodes: in this case, the number of transmissions is only determined by the size of the $p$-hop dominating set. Extra measurements will not cost extra transmissions because they are just different random linear combinations of the same signal coefficients already transmitted. In particular, the number of transmissions can be upper bounded\footnote{This is an upper bound because some paths may be counted multiple times.} by
\begin{align*}
n_{\text{TX}} \leq \sum_{j \in \D^p} \sum_{i=1}^p i \cdot \vert \Nc_j^i \vert
\end{align*}
being $\Nc_j^i$ the set of nodes whose shortest path to node $j$ is exactly $i$ edges long.
\item Insertion of new nodes: this strategy requires extra transmissions every time a node is inserted:
\begin{align*}
n_{\text{TX}} \leq \sum_{j \in \R} \sum_{i=1}^p i \cdot \vert \Nc_j^i \vert
\end{align*}
\end{itemize}

\section{Theoretical analysis}
\label{sec:theory}
In this section, we prove the main theoretical results, showing that the proposed method guarantees reconstruction and stability to noise. Then, guided by the results of the theoretical analysis we discuss the effects of system characteristics and free parameters.

\subsection{Restricted Isometry Property}
\label{sec:rip}
The proposed sampling method can be seen as a generalization of the method presented in \cite{eftekhari2015restricted}, where the authors use block diagonal matrices as sampling matrices. In our case, $\Phib$ is not a block-diagonal matrix, but we can expand the mathematical tools developed in \cite{eftekhari2015restricted} to our case. For this reason, the technique we use to prove that $\Psib$ satisfies the RIP is similar to the one proposed in \cite{eftekhari2015restricted}. 

Before presenting the main results of the section, we first introduce some preliminary definitions, referring to the construction of the sampling matrix $\Phib$ and sampling set $\R$ presented in Sec. \ref{sec:method}.
\begin{definition}
Given a graph $\G$, the graph-basis coherence $\mu(\Ub)$ of an orthobasis $\Ub\in\CC^{n\times n}$ is defined as follows
\begin{equation}
\label{eq:coherence}
\mu(\Ub)=\min\left(\sqrt{|\Ncb^*|}\lVert \Ub\rVert_{\scalebox{.6}{\textrm{MAX}}}, 1\right),
\end{equation}
where $|\Ncb^*|=\max_{i\in\R}|\Ncb_i|$ and $\lVert\Ub\rVert_{\scalebox{.6}{\textrm{MAX}}}=\max_{i,j} |\Ub_{ij}|$.
\end{definition}
\begin{definition}
\label{def:gmin}
The minimum node multiplicity $\gmin$ of $\Phib$ is defined as follows
\[
\gmin=\min_{k\in\V}g_k
\]
with $g_k$ defined as in \eqref{eq:g}.
\end{definition}
We can now introduce our main theoretical result.
\begin{theorem}
\label{thm:main}
Suppose $\Phib\in\RR^{m\times n}$ is defined as in \eqref{eq:def_phi} and suppose $\Ub\in\CC^{n\times n}$ is an orthobasis for $\CC^n$. Set $\Sigma_k=\{\xc\in\CC^n\mbox{ s.t. }\lVert \xc\rVert_0\le k,\ \lVert \xc\rVert_2=1 \}$. If $k\ge 1$ and 
\[
\label{eq:gmin_rip}
\gmin\ge c\delta^{-2}\mu^2 k \log^2 k \log^2 n,
\]
where $c$ is a positive constant, $\mu$ is the graph-basis coherence defined in \eqref{eq:coherence} and $0<\delta<1$, then there exists $\delta_k\le\delta<1$ such that, for all $\xb\in\RR^n$ with $\xc\in\Sigma_k$, it holds that
\[
(1-\delta_k) \le\lVert\Psib\xc\rVert_2^2\le(1+\delta_k)
\]
except with a probability of at most $O(n^{-\log n \log^2 k})$.
\end{theorem}
Before presenting the proof of Theorem \ref{thm:main}, we first recall an important lemma that will be used during the proof.
\begin{lemma}[Theorem 3.1 in \cite{krahmer2014suprema}]
\label{thm:gamma}
Let $\Ac\subset\CC^{m\times n}$ be a set of matrices, and let $\e$ be a random vector whose entries are i.i.d., zero-mean, unit-variance random variables with sub-Gaussian norm $\tau$. Set
\[
\begin{split}
&d_F(\Ac)=\sup_{\mathbf{T}\in\Ac}\lVert \mathbf{T}\rVert_F,\\  
&d_2(\Ac)=\sup_{\mathbf{T}\in\Ac}\lVert\mathbf{T}\rVert_2,
\end{split}
\]
and
\[
\begin{split}
&E_1=\gammas(\gammas+d_F(\Ac))+d_F(\Ac)d_2(\Ac),\\
&E_2=d_2(\Ac)(\gammas+d_F(\Ac)),\\
&E_3=d_2^2(\Ac)
\end{split},
\]
where the $\gamma_2$-functional of $\Ac$ is a geometrical property of $\Ac$, as defined in \cite{krahmer2014suprema}. Then, for $t>0$,
\begin{align*}
\log\PP\left[\sup_{\mathbf{T}\in\Ac}\Bigl|\lVert \mathbf{T}\e\rVert_2^2-\Ed[\lVert \mathbf{T}\e\rVert_2^2]\Bigr| \geq E_1 + t \right]\le \\ -c(\tau)\min\left(\frac{t^2}{E_2^2},\frac{t}{E_3}\right),
\end{align*}
where $c(\tau)$ is a positive constant that depends on $\tau$.
\end{lemma}
After having introduced this result, we can now present the proof of Theorem \ref{thm:main}.
\begin{proof}
First of all, we note that
\[
\begin{split}
\Ed[\lVert\Phib\xb\rVert_2^2]&=\Ed[\lVert\Psib\xc\rVert_2^2]=\xc^T\Ed[\Ub^H\Phib^H\Phib\Ub]\xc\\
&=\xc^T\Ub^H\Ed[\Phib^H\Phib]\Ub\xc=\xc^T\Ub^H\Ub\xc=\xc^T\xc=1,
\end{split}
\]
where the fourth equality above follows because the sampling set $\R$ is a superset of the dominating set $\D^p$ and $\varphi_{\pii j}\sim N\left(0,\frac{1}{g_j}\right)$ if $j\in\Ncb_i$.
In this way, we can define the restricted isometry constant $\delta_k$ as follows
\[
\delta_k=\sup_{\xc\in\Sigma_k}\Bigl|\lVert\Psib\xc\rVert_2^2-1\Bigr|.
\]
Given $\delta<1$, our objective is to show that, under the conditions described in the statement of the theorem, $\delta_k\le\delta$.

We define the diagonal matrix $\Gb\in\RR^{n\times n}$ as
\[
\Gb=\begin{bmatrix}\frac{1}{\sqrt{g_1}}& &&\\
&\frac{1}{\sqrt{g_2}}&&\\
&&\ddots&\\
&&&\frac{1}{\sqrt{g_n}}\\
\end{bmatrix},
\]
 in addition, given $\R$, we define the vector $\xb_{r_i}\in\RR^n$, where $r_i$ with $1\le i\le m$ is the $i$-th element of $\R$, as follows
 \[
 x_{{r_i},j}=\begin{cases} x_j &\mbox{ if }j\in\Ncb_{r_i}\\
 0 &\mbox{ if }j\notin\Ncb_{r_i}
 \end{cases}
 \]
 and the block diagonal matrix $\Abb\in\RR^{m\times n m}$ as
 \[
 \Abb=\begin{bmatrix}
 \xb_{r_1}^T\Gb& &&\\
 &\xb_{r_2}^T\Gb&&\\
 &&\ddots&\\
 &&&\xb_{r_m}^T\Gb\\
 \end{bmatrix}.
 \]
Then, we can observe that
\[
\begin{split}
\lVert\Phib\xb\rVert&=\sum_{i=1}^m|\phib^T_i\xb|^2=\sum_{i=1}^m|\phib^T_i\xb_{r_i}|^2=\sum_{i=1}^m|\xb_{r_i}^T\phib_i|^2\\
&=\sum_{i=1}^m|\xb_{r_i}^T\Gb\e_i|^2=\lVert \Abb\e \rVert_2^2,
\end{split}
\]
where $\phib_i\in\RR^n$ is the $i$-th row of $\Phib$, $\e_i\in\RR^{n}$ is a rescaled version of $\phib_i$ having standard normal distribution and $\e\in\RR^{mn}$ is made by stacking all the vectors $\e_i$. Thus, if we define the matrix set $\Acc=\{\Abb\ :\ \xc\in\Sigma_k\}$, we have completely expressed our problem in the setting of Lemma \ref{thm:gamma}. To apply this lemma we have first to compute $d_2(\Acc)$, $d_F(\Acc)$ and $\gammasb$. In order to compute these quantities, we first observe that
\begin{equation}
\label{eq:n_ax1}
\begin{split}
\lVert\Abb\rVert_2&=\lVert\Abb\Abb^T\rVert_2^{\frac{1}{2}}\le\frac{1}{\sqrt{\gmin}}\max_{1\le i\le m}\lVert \xb_{r_i} \rVert_2\\
&=\frac{1}{\sqrt{\gmin}}\max_{1\le i\le m}\lVert \Ub_{\Ncb_{r_i}}\xc \rVert_2\\
&\le \sqrt{\frac{|\Ncb^*|}{\gmin}}\max_{1\le i\le n}<\mathbf{u}_{i},\xc>\\
&\le \sqrt{\frac{|\Ncb^*|}{\gmin}}\max_{1\le i\le n}\lVert\mathbf{u}_{i}\rVert_{\infty}\lVert\xc\rVert_1=\sqrt{\frac{|\Ncb^*|}{\gmin}}\lVert\Ub\rVert_{\scalebox{.6}{\textit{MAX}}}\lVert\xc\rVert_1,
\end{split}
\end{equation}
where $\Ub_{\Ncb_{r_i}}\in\CC^{|\Ncb_i|\times n}$ is the matrix $\Ub$ restricted to the $\Ncb_{r_i}$ rows and $\mathbf{u}_{i}\in\CC^n$ is the $i$-th row of $\Ub$. On the other hand, we can also write that
\begin{equation}
\label{eq:n_ax2}
\begin{split}
\lVert\Abb\rVert_2&\le\frac{1}{\sqrt{\gmin}}\max_{1\le i\le m}\lVert \xb_{r_i} \rVert_2\le\frac{1}{\sqrt{\gmin}}\lVert \xb \rVert_2\\
&=\frac{1}{\sqrt{\gmin}}\lVert \xc \rVert_2\le\frac{1}{\sqrt{\gmin}}\lVert \xc \rVert_1.
\end{split}
\end{equation}
From \eqref{eq:n_ax1} and \eqref{eq:n_ax2} we arrive at
\[
\begin{split}
\lVert\Abb\rVert_2&\le\frac{1}{\sqrt{\gmin}}\min\left(\sqrt{|\Ncb^*|}\lVert\Ub\rVert_{\scalebox{.6}{\textit{MAX}}},1\right)\lVert \xc \rVert_1\\
&=\frac{\mu}{\sqrt{\gmin}}\lVert \xc \rVert_1
\end{split}
\]
Now, we can compute the quantities used in Lemma \ref{thm:gamma}. First, we can compute $d_2(\Acc)$ as follows
\[
d_2(\Acc)=\sup_{\Abb\in\Acc}\lVert\Abb\rVert_2\le\frac{\mu}{\sqrt{\gmin}}\sup_{\xc\in\Sigma_k}\lVert \xc \rVert_1\le\mu\sqrt{\frac{k}{\gmin}},
\]
Second, we have that
\[
\begin{split}
d_F(\Acc)&=\sup_{\Abb\in\Acc}\sqrt{\mbox{tr}(\Abb\Abb^T)}=\sup_{\Abb\in\Acc}\sqrt{\sum_{i=1}^m\sum_{j=1}^n\frac{1}{g_j}x^2_{r_i,j}}\\
&=\sup_{\Abb\in\Acc}\lVert\xb\rVert_2=\sup_{\Abb\in\Acc}\lVert\xc\rVert_2=1.
\end{split}
\]
In order to compute $\gammasb$, we apply the same procedure used in \cite{eftekhari2015restricted} and we obtain that
\[
\gammasb\le c\mu\sqrt{\frac{k}{\gmin}}\log k \log n.
\]
Finally, we can now apply Lemma \ref{thm:gamma}. Given $\delta$ such that $0<\delta<1$ and assuming that $\gmin\ge c\delta^{-2}\mu^2 k \log^2 k \log^2 n$, we obtain that
\[
\begin{split}
    E_1&=\gammas(\gammas+d_F(\Ac))+d_F(\Ac)d_2(\Ac)\\
    &\le c\mu\sqrt{\frac{k}{\gmin}}\log k \log n\left(c\mu\sqrt{\frac{k}{\gmin}}\log k \log n+1\right)\\
    &+\mu\sqrt{\frac{k}{\gmin}}\le \delta(\delta+1)+\frac{\delta}{c\log k\log n}\le c_1\delta,
\end{split}
\]
\[
\begin{split}
    E_2&=d_2(\Ac)(\gammas+d_F(\Ac))\\
    &\le \mu\sqrt{\frac{k}{\gmin}}\left(c\mu\sqrt{\frac{k}{\gmin}}\log k \log n +1\right)\\
    &\le \frac{\delta}{c\log k \log n}(\delta+1)\le c_2 \frac{\delta}{\log k \log n},
\end{split}
\]
\[
\begin{split}
    E_3&=d_2^2(\Ac)\le\frac{k\mu^2}{\gmin}\le c_3 \frac{\delta^2}{\log^2 k\log^2 n}. \\
\end{split}
\]
Then, the tail bound in Lemma 1 can be written as
\[
\begin{split}
&\log\left(\PP\left[\sup_{\xc\in\Sigma_k}\Bigl|\lVert\Phi\xb\rVert_2^2-1\Bigr|>c_1\delta + t \right] \right)\\
&\le -c(\tau)\min\left(t^2\delta^{-2}\log^2 k \log^2 n, t\delta^{-2}\log^2 k \log^2 n\right),
\end{split}
\]
If we set $t=\delta$, we obtain
\[
\log\left(\PP\left[\sup_{\xc\in\Sigma_k}\Bigl|\lVert\Phi\xb\rVert_2^2-1\Bigr|>\delta \right] \right)\le -c\log^2 k\log^2 n,
\]
which completes the proof.
\end{proof}
We point out that Theorem \ref{thm:main} does not require $\Phib$ to be full row-rank. However, adding measurements that are linear combination of the previous ones is not efficient because these new measurements do not add new information. If we consider a sampling strategy without repeated nodes like the second one presented in Section \ref{sec:samp_strat}, we have the theoretical guarantee that $\Phib$ is full row-rank, as shown in the following theorem.
\begin{theorem}
\label{thm:rank}
Given a sampling set $\R$, 
where there are no repeated copies of the same node, set $\Phib = \bar{\Ab}_\R \odot \Xib$, where $\Xib$ is a random matrix with independent entries drawn from a continuous distribution. Then, $\Phib$ is a full row-rank matrix with probability 1.
\end{theorem}
 \begin{proof}

Let $\phib_{\pii}$ be the row of $\Phib$ corresponding to node $i$. Given a subset $\W_i$ of the set $\R-\{i\}$, we define the union of the supports of $\{\phib_{\pij}\}_{j\in\W_i}$ as the set $\C_i=\bigcup_{j\in\W_i}\Ncb_j$.  Since the nonzero entries of $\Phib$ are drawn independently from a continuous distribution, we have that $\phib_{\pii}$ is linearly dependent on $\{\phib_{\pij}\}_{j\in\W_i}$ with probability 1 if and only if there exists a subset $\W_i\subseteq\R-\{i\}$ such that $\Ncb_i\subseteq\C_i$ and $|\W_i|\ge|\C_i|$, i.e., $\mbox{span}(\{\phib_{\pij}\}_{j\in\W_i})=\RR^{|\C_i|}$. Otherwise, $\phib_{\pii}$ is linearly dependent on $\{\phib_{\pij}\}_{j\in\W_i}$ with probability 0 if for any $\W_i \subseteq \R-\{i\}$, $\Ncb_i \not\subset \C_i$ or $|\W_i|<|\C_i|$. Then, every $j \in \W_i$ also implies $j \in \C_i$ because of self-loops and the following holds: $|\C_i| \geq |\W_i|$. When $i \in \C_i$ because of edges from a node in $\W_i$, it makes $|\C_i| \geq |\W_i|+1$, thus implying that $\phib_{\pii}$ is linearly dependent on $\{\phib_{\pij}\}_{j\in\W_i}$ with probability 0. Otherwise, $i \not\in \C_i$ implies $\Ncb_i \not\subset \C_i$ thanks to the self-loop on $i$, thus also implying that $\phib_{\pii}$ is linearly dependent on $\{\phib_{\pij}\}_{j\in\W_i}$ with probability 0 and concluding the proof. As a remark, we notice that this proof would not be true with repeated nodes because the self-loop on $i$ could already be included in the supports of $\{\phib_{\pij}\}_{j\in\W_i}$ thus making the statement $|\C_i| \geq |\W_i|+1$ for the $i \in \C_i$ case false.

 
 \end{proof}
 
 \subsection{Uniqueness of solution with known support}
 
 We now provide the conditions enabling the proposed method to achieve perfect reconstruction when the signal support is known a priori.
 
 \begin{theorem}
 \label{thm:recovery}
 Given a graph $\G$ and an orthobasis $\Ub\in\CC^{n\times n}$, we consider a graph signal $\xb\in\RR^n$ having a $k$-sparse representation under the basis $\Ub$ with known support $\Sx$, i.e., $\xb = \Ub_{|\Sx}\hat{\xb}_\Sx$ and measurements $\yb = \Phib \xb = \Psib \hat{\xb} = \Psib_{|\Sx} \hat{\xb}_\Sx$. Then, perfect reconstruction is achieved with probability 1 when 
 $\gmin\ge k$.
 \end{theorem}
 \begin{proof}
 In order to show that we can achieve perfect reconstruction, we have to prove that, under the conditions described in the statement of the theorem, the matrix $\Psib_{|\Sx}$ is full rank $\forall\,\Sx$ such that $|\Sx|=k$. Therefore, the $k$ columns of $\Psib_{|\Sx}$ should be linearly independent. Let $\psib_i$ be the $i$-th column of $\Psib_{|\Sx}$, 
 given any $\psib_j$ with $j\neq i$, we point out that
 \[
 \Ed[\psib_j^T\psib_i]=\Ed[\ub_j^T\Phib^T\Phib\ub_i]=\ub_j^T\Ed[\Phib^T\Phib]\ub_i=\ub_j^T\ub_i=0,
 \]
 where $\ub_i$ is the $i$-th column of $\Ub$. Thus, we can say that $\psib_j$ and $\psib_i$, with $i\neq j$, are uncorrelated. We also recall that the entries of $\psib_i$, with $1\le i\le n$, are linear combinations of Gaussian random variables, then $\psib_i$ is a random vector with Gaussian entries. Moreover, it is easy to see that any linear combinations of $\psib_i$ and $\psib_j$ is still a random vector with Gaussian entries. Therefore, we can say that the entries of $\psib_i$ and $\psib_j$ are jointly Gaussian distributed \cite{Gut_probability}. Then, since $\psib_i$ and $\psib_j$ are uncorrelated and their entries are jointly Gaussian distributed, from Theorem 4.5 in \cite{hogg2014probability}  we can say that the entries of $\psib_i$ and $\psib_j$ are independent.
 
 Let $\Sx_i$ be the support of $\psib_i$. Then, since we have shown that $\forall\, j\in\Sx$ with $j\neq i$ $\psib_i$ and $\psib_j$ are pairwise independent random vectors whose entries are drawn from a continuous distribution, $\psib_i$ is linearly dependent with probability 1 if and only if there exists a subset $\C_i\subseteq\Sx-\{i\}$ with $\Sx_{\C_i}=\bigcup_{j\in\C_i}\Sx_j$ such that $\Sx_i\subseteq\Sx_{\C_i}$ and $|\C_i|\ge|\Sx_{\C_i}|$, i.e., $\mbox{span}(\{\psib_j\}_{j\in\C_i})=\RR^{|\Sx_{\C_i}|}$. Otherwise, $\psib_i$ is linearly dependent with probability 0 if for any subset $\C_i\subseteq\Sx-\{i\}$ with $\Sx_{\C_i}=\bigcup_{j\in\C_i}\Sx_j$, $\Sx_i \not\subset \Sx_{\C_i}$ or $|\C_i|<|\Sx_{\C_i}|$. Then, we can observe that
 \[
 k-1\ge|\C_i|\ge|\Sx_{\C_i}|\ge|\Sx_i|.
 \]
 In order to compute $|\Sx_i|$, we recall that $\psib_i=\Phib\ub_i$. Since $\ub_i$ is an orthonormal vector with respect to $\ub_j$, we can say that there exists at least one nonzero entry $l$ of the vector $\ub_i$ and the corresponding $l$-th column of $\Phib$ has $g_l$ nonzero entries by construction. Therefore, we obtain that $|\Sx_i|\ge g_l\ge\gmin$, which completes the proof.
 
 \end{proof}
 
 \subsection{Discussion}
As we have shown in Section \ref{sec:rip}, if we consider the second sampling strategy presented in Section \ref{sec:samp_strat}, we have the theoretical guarantee that new measurements are linearly independent from the previous ones and thus always informative. Instead, the strategy presented in Section \ref{sec:samp_strat} that repeats dominating nodes can only have this guarantee if at every insertion the rank of the resulting matrix is checked to avoid the presence of subgraphs with linearly dependent measurements. 

It is also important to underline the fundamental role of the minimum node multiplicity $\gmin$ in Theorem \ref{thm:main}. As we can see from Definition \ref{def:gmin}, the parameter $\gmin$ represents the minimum number of times a node of the graph appears in the measurements and it is subject to the lower bounds in Theorem \ref{thm:main} and \ref{thm:recovery} in order to satisfy the RIP and perfect reconstruction conditions. For these reasons, the proposed sampling strategies aim to find the sampling node that provides the fastest increase of $\gmin$, as shown in Eq. \eqref{eq:node_insert}. We also point out that our method is not affected if the graph is disconnected. In fact, the sampling strategy based on $\gmin$ guarantees that all the connected components are fairly sampled.

Furthermore, we point out that Theorem \ref{thm:recovery} provides a sufficient condition to guarantee perfect reconstruction. However, in some specific cases the bound might be not very tight and it could be impossible to satisfy the condition  $\gmin\ge k$. In such cases, Theorem \ref{thm:recovery} becomes uninformative. For example, if we consider a disconnected graph, it may happen that $\gmin$ is always lower than $k$, even if we increase the number of measurements. In this specific case, we can easily obtain a more informative lower bound by considering each connected component separately. Therefore, even if the signal is localized in just a few components, we can still reconstruct the signal given a sufficiently large number of measurements.

It is also interesting to highlight the connection between the theoretical results presented in this section and some well-known results of compressed sensing. In particular if we consider a fully connected graph, we have that $\Phib$ is a dense matrix and $\gmin$ is equal to $m$; then the condition in Theorem 3 becomes $m\ge k$, which is a known result for compressed sensing. Instead, if we consider a disconnected graph where each component has the same number of nodes and is fully connected, it can be shown that $\Phib$ is a random block-diagonal matrix and the RIP property presented in Theorem \ref{thm:main} becomes equal to the result presented in \cite{eftekhari2015restricted}.

Moreover, it is worth noting that, if we consider a dense sparsity basis $\Ub$, the condition of Theorem \ref{thm:recovery} on $\gmin$ can be replaced by requiring $m\ge k$, because in this case we have the guarantee that $\Psib$ is dense.

Another important parameter that appears in Theorem \ref{thm:main} is the graph-basis coherence $\mu(\Ub)$. It can be easily verified that
\[
\sqrt{\frac{\Ncb^*}{n}}\le\mu(\Ub)\le 1.
\]
If we consider an undirected cycle graph and take as orthonormal basis $\Ub$ its GFT basis (which in this case corresponds to the Fourier basis), it can be shown that $\mu(\Ub)=\sqrt{\frac{2}{n}}$ \cite{eftekhari2015restricted}.


\section{Experiments}
\label{sec:experiments}

This section compares the performance of the proposed method with some other sampling methods for signals defined over graphs present in the literature.

\subsection{Experimental setup}

We consider a signal that is $k$-sparse in the GFT domain, defined as the eigenvectors of the normalized graph Laplacian $\Lb$.
Depending on the particular experiment, the frequency support can either be the first $k$ coefficients (bandlimited model) or $k$ randomly chosen frequencies (random model).
Tests are performed using the following graphs as implemented in the Matlab Graph Signal Processing Toolbox \cite{perraudin2014gspbox}: Minnesota, Sensor, Community. The dominating set is computed using a greedy algorithm that at every step adds the node with the largest degree without any neighbors already in the dominating set. 

We perform sampling and reconstruction experiments with both known and unknown signal support. When the support is known, we compare the proposed method with uniform random sampling \cite{Puy_random}, weighted random sampling \cite{Puy_random}, the MinPinv \cite{Tsitsvero_uncertainty} greedy approximation to the optimal sampling and reconstruction problem and the method based on the localization operator \cite{sakiyama_2017} (using a Chebyshev polynomial approximation of order 10 for the kernel). Notice that weighted random sampling and MinPinv require knowledge of the signal support, so we only test uniform random sampling and the localization operator when the support is unknown.
Also notice that with respect to the other techniques, the proposed method with the repeated dominating nodes does not reduce locality when increasing the number of measurements: only $\vert \D^p \vert \leq n$ nodes are aggregation points. Also, the hop parameter $p$ decreases as $m$ increases. All the results are averaged over 1000 runs.
We also provide results about the running time of the sampling procedure (without signal reconstruction) for each method.

Finally, we do not show any reconstruction results for the successive local aggregations method of \cite{Marques_successive}. This method suffers from numerical instability even for moderately-sized graphs so that the reconstruction error quickly explodes for the sizes used in our experiments. Table \ref{table:condition_Marques} highlights this fact by reporting the high value of the condition number of the matrix $\Psib$ obtained with such method, for the Erd\H{o}s-R\'{e}nyi graph with edge probability $p_e$. As mentioned in Section \ref{sec:bkg_related} this is due to the poor conditioning of the Vandermonde matrix constructed by the sampling procedure.

\begin{table}[t]
    \centering
    \caption{Condition number of sampling matrix Erdos-Renyi graph, $n=100$, $k=10$}
    \label{table:condition_Marques}
    \begin{tabular}{c|c|c|c|c|c|c|c}
    \multirow{2}{*}{$\mathbf{p_e}$} & \multirow{2}{*}{\bf{Method}} & \multicolumn{6}{|c}{\bf{Measurements}} \\
    & & 10 & 20 & 30 & 40 & 50 & 100 \\
    \hline
    & & & & & & & \\[-6pt]
    \multirow{2}{*}{$0.2$} & \cite{Marques_successive} & $10^{8}$ & $10^{17}$ & $10^{19}$ & $10^{20}$ & $10^{22}$ & $10^{28}$ \\
    & Proposed & $197.2$ & $3.8$ & $2.5$ & $2.2$ & $2.1$ & $1.6$\\
    \hline
    \multirow{2}{*}{$0.5$} & \cite{Marques_successive} & $10^{10}$ & $10^{17}$ & $10^{18}$ & $10^{19}$ & $10^{19}$ & $10^{25}$ \\
    & Proposed & $36.7$ & $4.5$ & $3.1$ & $2.3$ & $2.0$ & $1.9$\\
    \hline
    \multirow{2}{*}{$0.8$} & \cite{Marques_successive} & $10^{12}$ & $10^{17}$ & $10^{17}$ & $10^{17}$ & $10^{18}$ & $10^{21}$ \\
    & Proposed & $42.6$ & $3.6$ & $3.0$ & $2.5$ & $2.3$ & $1.5$
    \end{tabular}
\end{table}

\subsection{System behavior}

We show the size of the $p$-hop dominating set as returned by the greedy algorithm for the Minnesota and Sensor graphs. From Fig. \ref{fig:dominating} it can be noticed that the Minnesota graph has a particularly large dominating set due to the low degree of many nodes. However, the dominating set quickly shrinks once multiple hops are considered.

\begin{figure}[t]
\centering
\subfigure[Minnesota]{
\includegraphics[width=0.48\columnwidth]{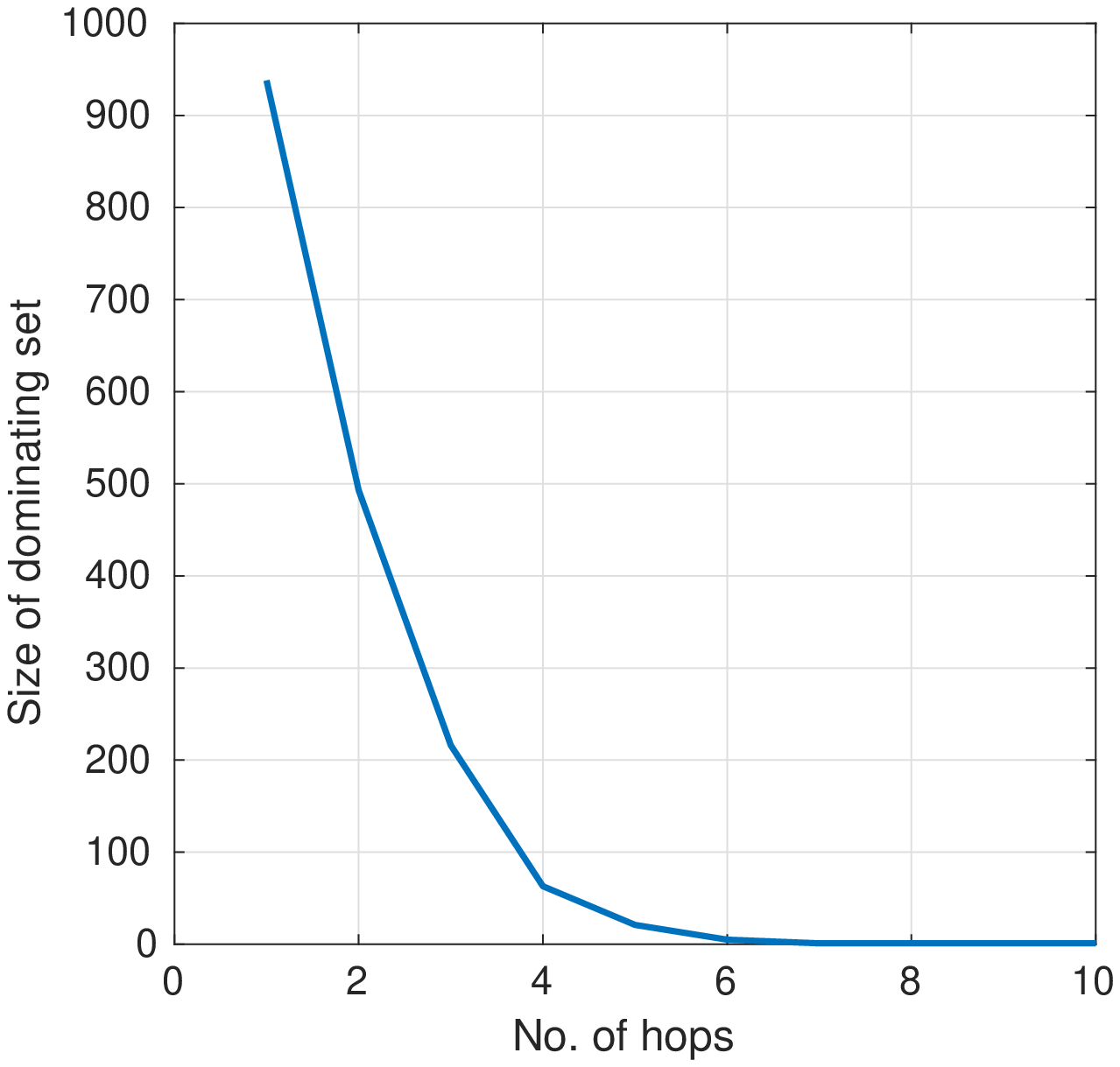}
}\hspace*{-0.1cm}
\subfigure[Sensor]{
\includegraphics[width=0.47\columnwidth]{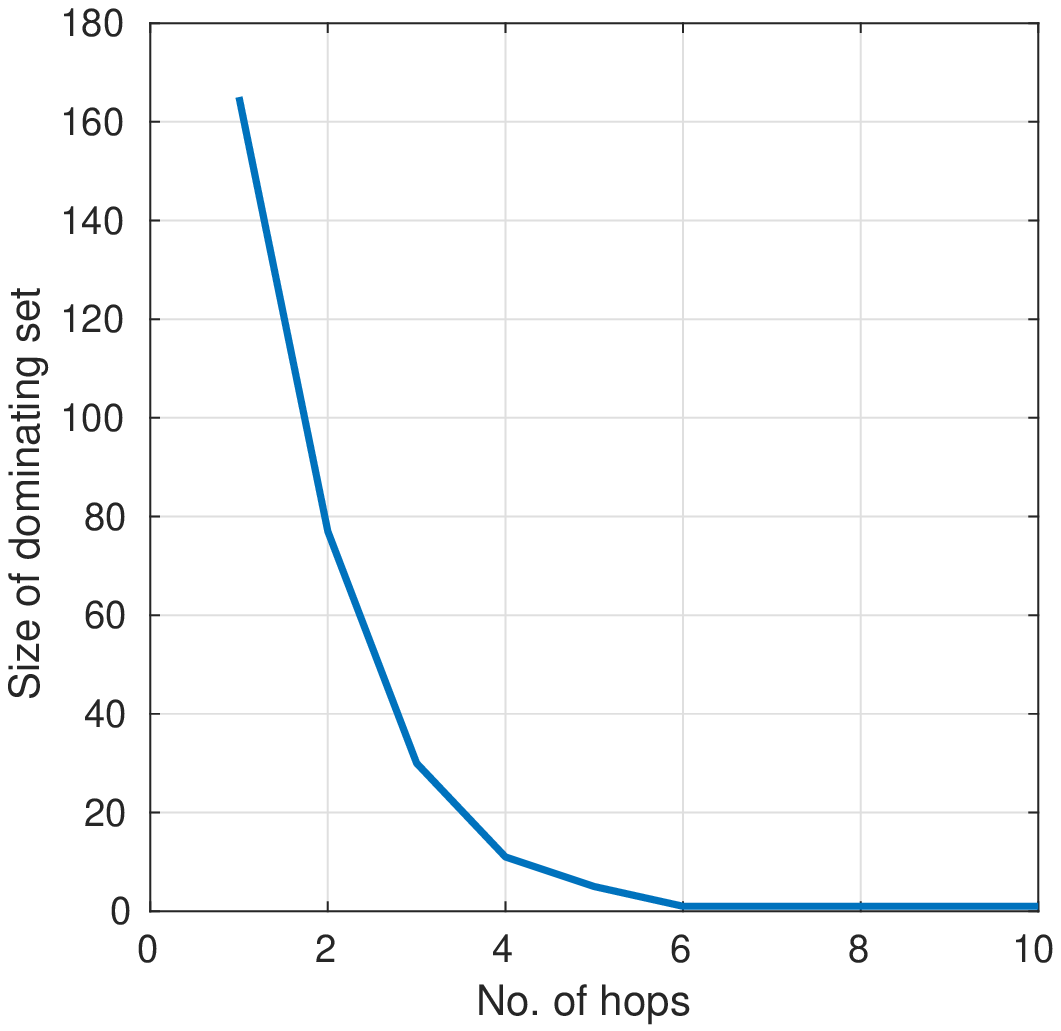}
}
\caption{Size of dominating set as function of number of hops used to consider two nodes a neighbors ($n=2642$ for Minnesota, $n=1000$ for Sensor).}
\label{fig:dominating}
\end{figure}



\subsection{Reconstruction: known support}

In this experiment, the signal support in the frequency domain is known. However, the true sparse signal is corrupted by additive white Gaussian noise with standard deviation $\sigma$, i.e., measurements are computed as $\yb = \Phib (\xb + \nb)$. Least-squares reconstruction is performed by computing the pseudoinverse of the product between the sampling operator and the basis, restricted to the the columns in the support:
\begin{align*}
\hat{\xb}_\Sx^* &= \left( \Psib_{\vert\Sx} \right)^+ \yb \\
\xb^* &= \Ub \hat{\xb}^*
\end{align*}
Finally, we estimate the mean squared error (MSE) between the reconstructed signal $\xb^*$ and the original noiseless signal $\xb$ as:
\begin{align*}
\text{MSE} = 10\log_{10}\left( n^{-1} \Vert \xb^* - \xb \Vert^2 \right) .
\end{align*}

Fig. \ref{fig:supp_sensor_community} shows the reconstruction MSE for the Sensor and Community graphs, as function of the number of measurements for a fixed noise standard deviation. We notice that the MinPinv algorithm \cite{Tsitsvero_uncertainty} performs best, as expected since it approximates the optimal solution to the sampling problem which would minimize the reconstruction MSE. Its drawback is that it is computationally expensive and not very flexible, as it can only be applied if the support is known. The proposed method is typically close to MinPinv, especially for a lower number of measurements where it outperforms methods based on uniform or weighted random sampling and it is on par with or slightly better than the approach using the localization operator. We notice that the proposed method has a different error floor with respect to other methods for a large number of measurements due to the construction of the sampling matrix, having a condition number slightly larger than 1. However, it is not particularly interesting as this floor occurs for low values of MSE. We also notice that the proposed method works well even with the random model, when uniform random sampling and the localization operator exhibit poor performance due to their implicit assumptions on bandlimited signals.

\begin{figure}[t]
\centering
\subfigure[Sensor]{
\includegraphics[width=0.49\columnwidth]{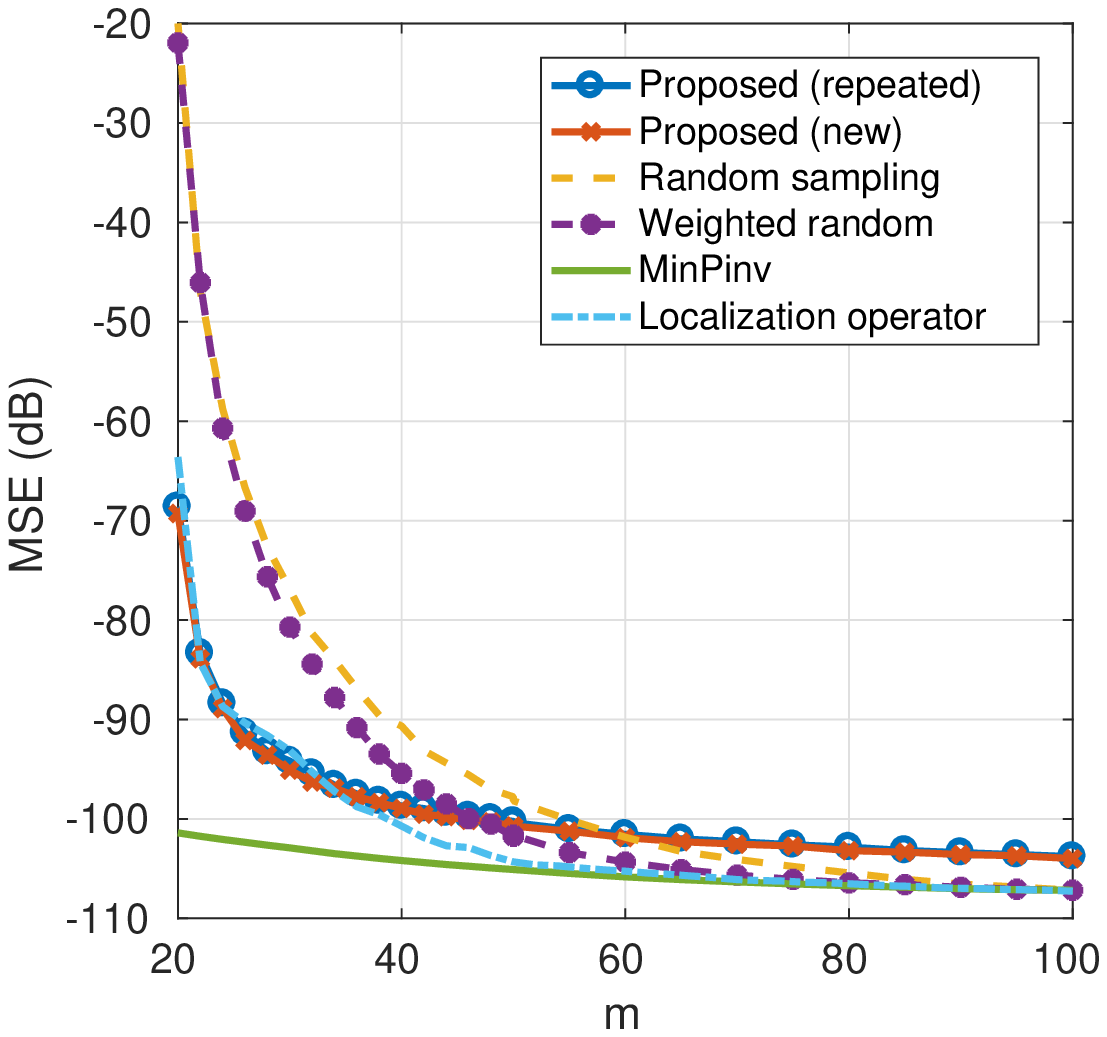}
}\hspace*{-0.1cm}
\subfigure[Community]{
\includegraphics[width=0.47\columnwidth]{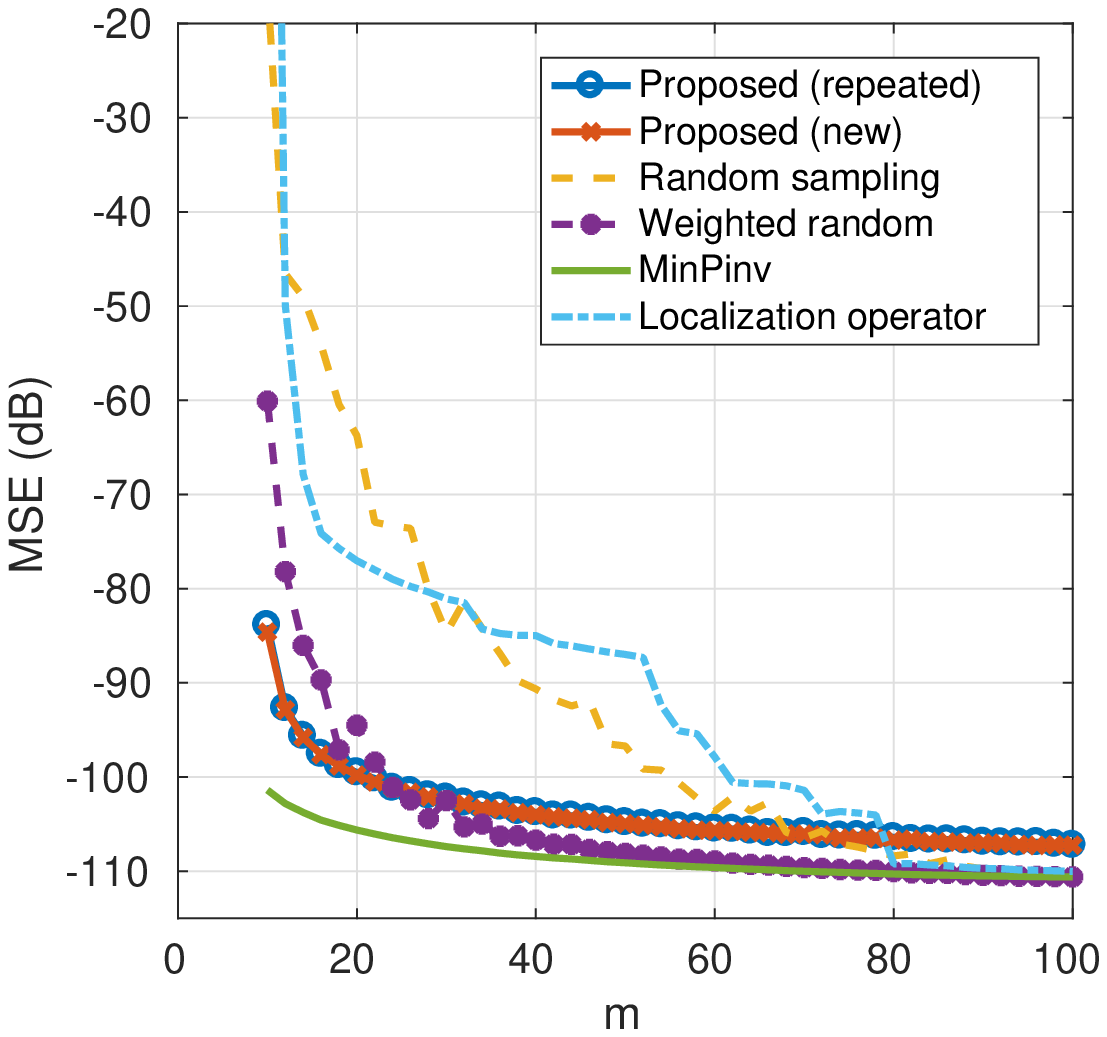}
}
\caption{Reconstruction MSE as function of number of measurements. Sensor graph, $n=100$, $k=20$, bandlimited model. Community graph, $n=100$, $k=10$, random model. Noise standard deviation $\sigma=10^{-5}$.}
\label{fig:supp_sensor_community}
\end{figure}

Fig. \ref{fig:supp_minnesota} reports the results for the Minnesota graph. It is important to notice that the proposed method is stable in presence of noise, as confirmed by Fig. \ref{fig:supp_minnesota_vs_noise} where the slope matches the ones of the quasi-optimal MinPinv. We also observe that the two node selection strategies are fairly equivalent. Finally, we note that all the tests achieved perfect reconstruction with $m=k$ when no noise was added.

\begin{figure}[t]
\centering
\subfigure[MSE vs. $m$]{
\includegraphics[width=0.49\columnwidth]{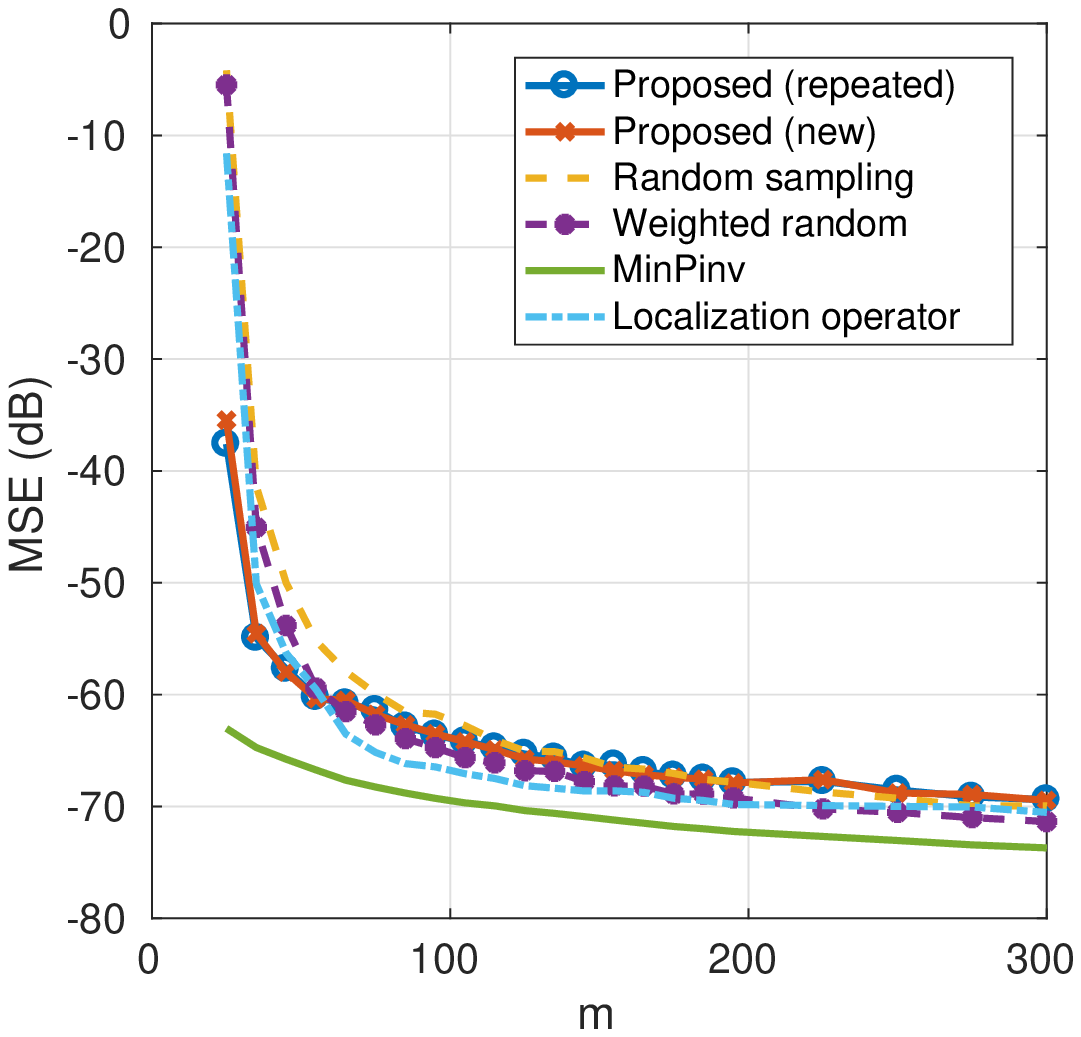}
}\hspace*{-0.1cm}
\subfigure[MSE vs. $\sigma$]{
\includegraphics[width=0.495\columnwidth]{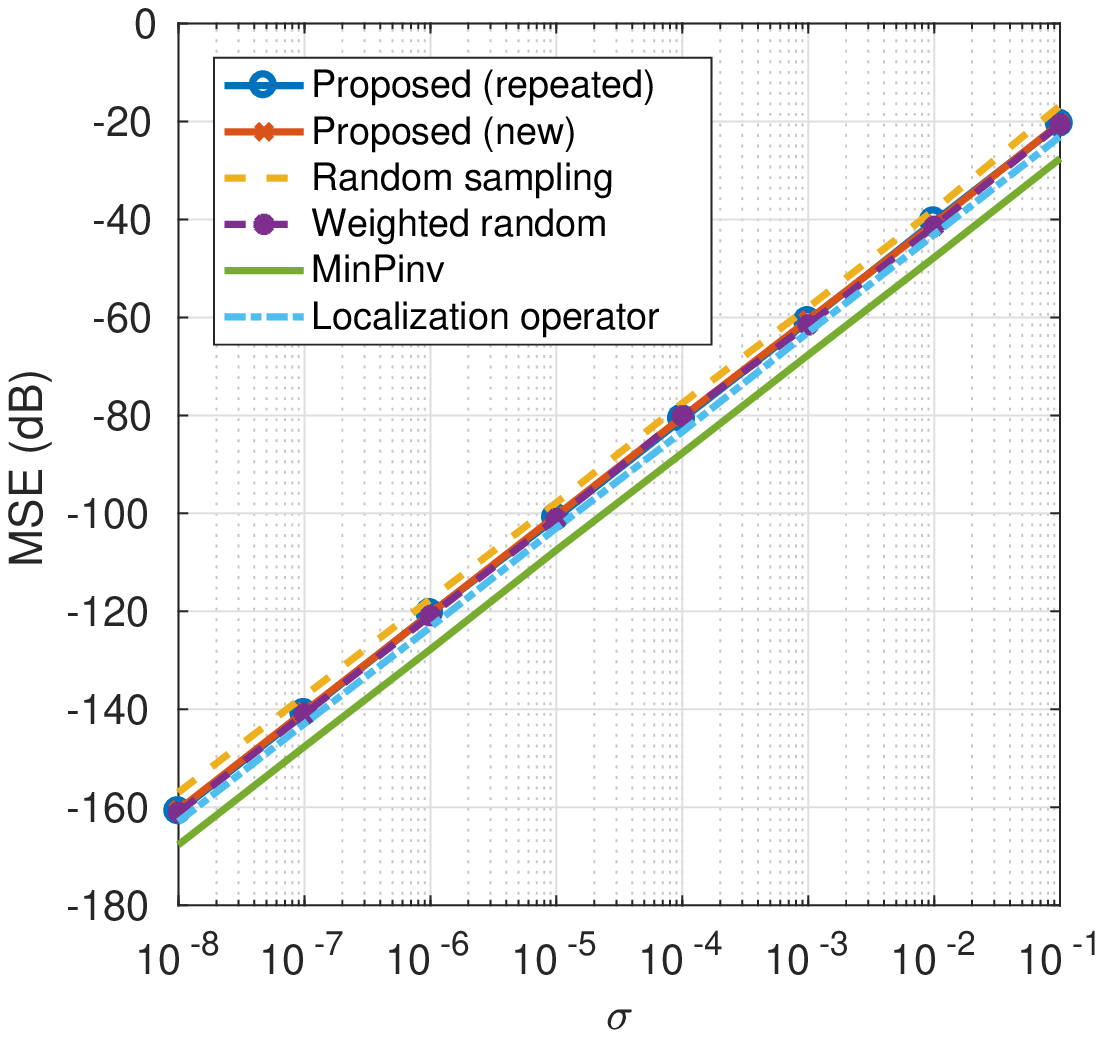}
\label{fig:supp_minnesota_vs_noise}
}
\caption{Reconstruction MSE as function of number of measurements $m$ and noise standard deviation $\sigma$. Minnesota graph, $n=2642$, $k=25$, bandlimited model. (a) $\sigma=10^{-3}$, (b) $m=65$.}
\label{fig:supp_minnesota}
\end{figure}

\subsection{Reconstruction: unknown support}

This experiment reconstructs signals with unknown frequency support by using $\ell_1$ minimization, as defined in Eq. \eqref{eq:l1min}. No noise is added in this experiment and we measure the probability of perfect reconstruction, defined as a MSE lower than -40 dB, as a function of the number of measurements. 

Fig. \ref{fig:nosupp_minnesota_sensor} shows the probability of perfect reconstruction for the Minnesota and Sensor graphs. It can be noticed that the phase transition of the proposed method occurs at a lower number of measurements with respect to uniform random sampling. Fig. \ref{fig:nosupp_community} reports the same experiment on the community graph. Notice how random sampling performs well for a bandlimited signal but degrades with respect to the proposed method when the support is chosen at random. The method based on the localization operator performs worse on this graph both under the bandlimited and random models. 

Finally, we test the effect of the graph-basis coherence on the reconstruction performance. Theorem \ref{thm:main} predicts that a higher value of $\mu$ requires a higher $g_\text{min}$ which means a higher number of measurements for a given graph. We constructed an experiment with two graphs with the same number of nodes ($n=100$) and same size of the dominating set ($\vert \D \vert=50$), but with drastically different values of the graph-basis coherence with the GFT basis. The first graph is a 2D grid, showing a low coherence, while the second is a small-world graph with high coherence. Fig. \ref{fig:coherence} shows that the phase transition occurs for a larger number of measurements when the coherence is high. We also test the impact of a weighted matrix used to compute the GFT basis. In particular, we draw a random geometric graph with a fixed number of nodes ($n=100$) on the unit 2D square and an edge is placed if the Euclidean distance between the nodes is below a certain threshold (equal to 0.2). We then repeat the same experiment as before of reconstruction with unknown support, comparing the probability of perfect reconstruction with the GFT as sparsity basis, computed from the normalized Laplacian of the binary adjacency matrix or from the weighted adjacency matrix where the weights of the nonzero entries are $w_{ij}=e^{-d_E(i,j)}$, being $d_E$ the Euclidean distance between two nodes in the 2D space. Fig. \ref{fig:weighted_gft} shows that there is no substantial difference in performance between the two bases.

\begin{figure}[t]
\centering
\subfigure[Minnesota]{
\includegraphics[width=0.48\columnwidth]{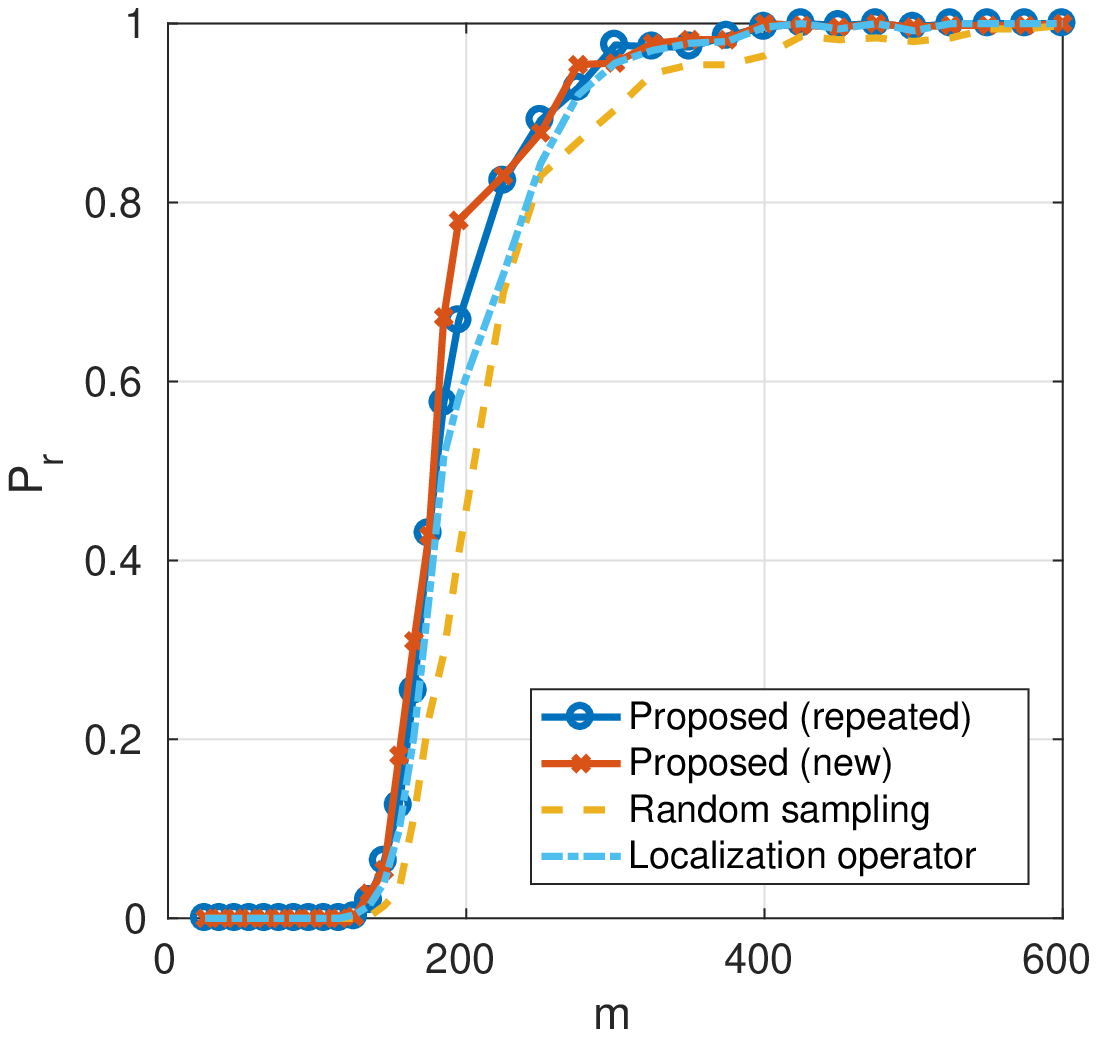}
}\hspace*{-0.1cm}
\subfigure[Sensor]{
\includegraphics[width=0.48\columnwidth]{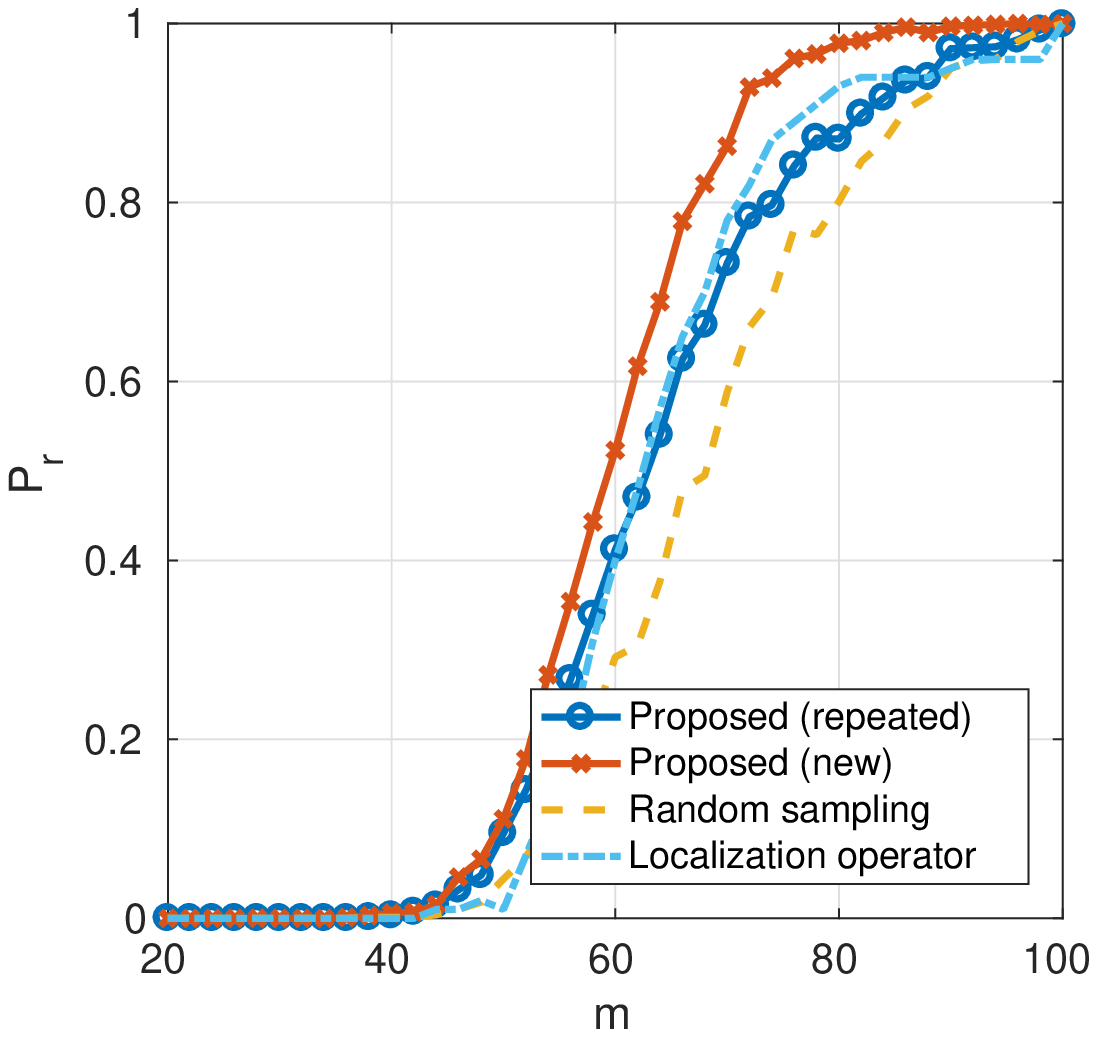}
}
\caption{Support unknown: probability of perfect recovery as function of number of measurements $m$. Minnesota graph, $n=2642$, $k=25$, bandlimited model. Sensor graph, $n=100$, $k=20$, random model.}
\label{fig:nosupp_minnesota_sensor}
\end{figure}

\begin{figure}[t]
\centering
\subfigure[Community (bandlimited)]{
\includegraphics[width=0.48\columnwidth]{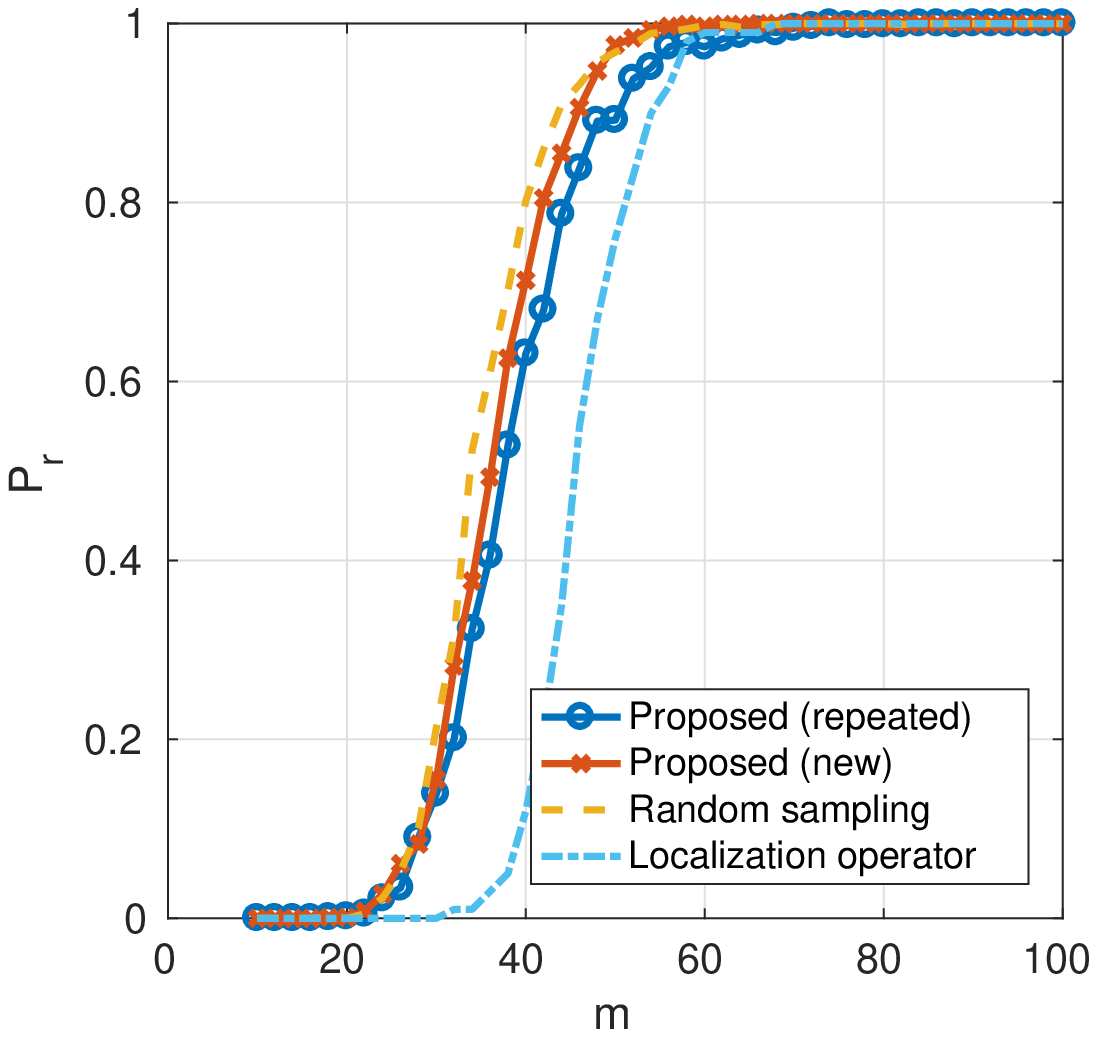}
}\hspace*{-0.1cm}
\subfigure[Community (random)]{
\includegraphics[width=0.48\columnwidth]{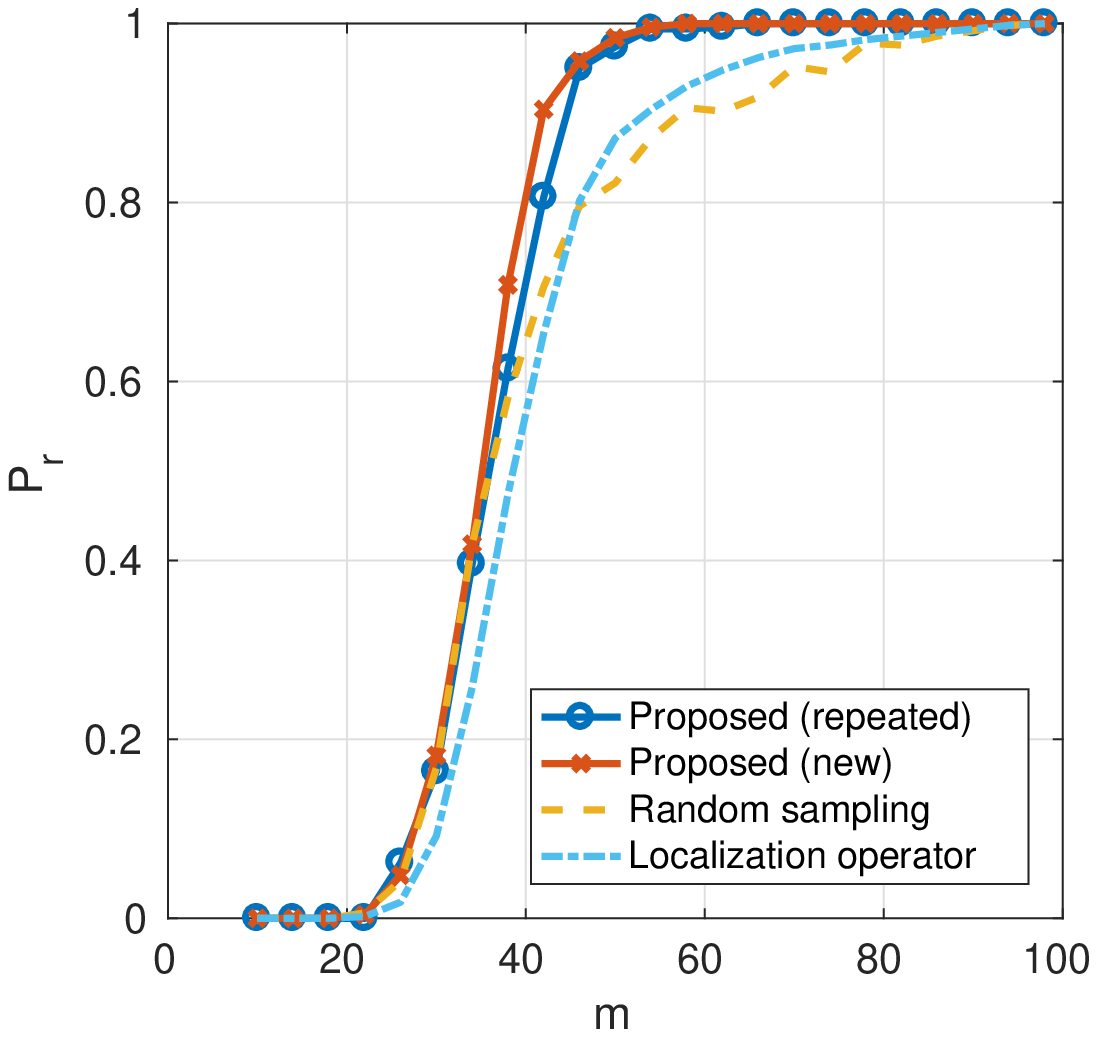}
}
\caption{Support unknown: probability of perfect recovery as function of number of measurements $m$. Community graph, $n=100$, $k=10$, (a) bandlimited model, (b) random model. }
\label{fig:nosupp_community}
\end{figure}

\begin{figure}[t]
\centering
\includegraphics[width=0.55\columnwidth]{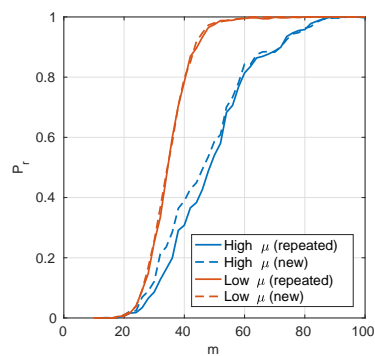}
\caption{Support unknown: probability of perfect recovery as function of number of measurements $m$. $n=100$, $k=10$, bandlimited model. Low $\mu$: 2D grid graph ($\mu = 0.68$ at $m=60$), High $\mu$: small-world graph ($\mu = 1$ at $m=60$)}
\label{fig:coherence}
\end{figure}

\begin{figure}[t]
\centering
\includegraphics[width=0.55\columnwidth]{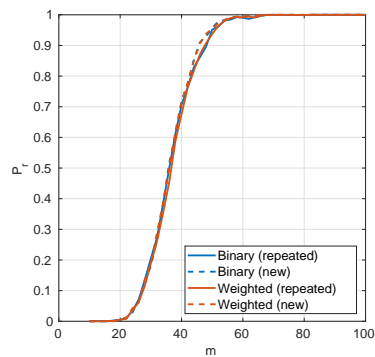}
\caption{Support unknown: probability of perfect recovery as function of number of measurements $m$. $n=100$, $k=10$, bandlimited model. Random geometric graph in the 2D square with Euclidean distance threshold 0.2. Basis is GFT from binary weights or exponential weights.}
\label{fig:weighted_gft}
\end{figure}

\subsection{Computational complexity}
Table \ref{table:complexity} reports the running time for all the tested methods (except uniform random sampling, whose complexity is trivial) on the community graph with $n=2000$ nodes, a signal with sparsity $k=20$ on a random support and for different number of measurements. The times include all the operations needed to produce the measurements, thus including the computation of the eigenvectors of the Laplacian for the methods requiring them, or the dominating set for the proposed method. Notice that for the proposed method, the 1-hop dominating set has size $\vert \D \vert = 163$, the 2-hop dominating set has size $\vert \D^2 \vert = 48$ and the 3-hop dominating set has size $\vert \D^2 \vert = 5$. The methods were tested in MATLAB on a PC with an Intel Xeon E5620 CPU and 24GB of RAM. It can be noticed that the proposed method is faster, except when the number for measurements is very large and the repetition strategy is used. As discussed in Sec. \ref{sec:samp_strat} this is almost entirely due to the check on linear independence of the new measurement. On the other hand, the runtime of the weighted random method is dominated by the computation of the eigenvectors of the Laplacian, and in the case of MinPinv by the partial SVDs. The method based on the localization operator scales well thanks to the use of polynomial filters and is close to the proposed method.

\begin{table}[]
    \centering
    \caption{Sampling runtime (seconds)}
    \begin{tabular}{c|ccccc}
    \multirow{2}{*}{$m$} & \multirow{2}{*}{\textbf{Weighted}} & \multirow{2}{*}{\textbf{MinPinv}} & \multirow{2}{*}{\textbf{Localization}} & \multirow{2}{*}{\textbf{Proposed}} & \multirow{2}{*}{\textbf{Proposed}}\\[3pt]
    & \textbf{random} & & \textbf{operator} & \textbf{(repeated)} & \textbf{(new)}\tabularnewline
    \hline
    20 & 5.88 & 8.25 & 2.06 & \textbf{1.29} & 1.71\tabularnewline
    100 & 5.89 & 31.56 & 3.98 & \textbf{1.25} & 1.87\tabularnewline
    500 & 5.94 & 203.17 & 14.40 & 16.49 & \textbf{4.99}\tabularnewline
    \end{tabular}
    \label{table:complexity}
\end{table}

\subsection{Application: data gathering in sensor networks}
Wireless sensor networks are made of multiple, typically inexpensive sensors, tasked with acquiring a signal in a certain area. Such sensors usually have very limited computational capabilities and batteries, so power management is critical. We will consider the same scenario of data gathering presented in \cite{wsn_datagather}, in which a signal is sampled by means of a few measurements that are transmitted only by small portion of nodes to a remote base station. \cite{wsn_datagather} also uses an approach based on CS but it is less flexible than the one presented in this paper as it results in a block-diagonal sensing matrix. The setup of the experiments considers $n$ sensors randomly positioned inside the unit square and a remote base station at a distance equal to 5 times the side of the square. Sensors incur a transmission cost due to the consumed power being equal to the square of the transmission distance: $P_{ij}=d_{ij}^2$. The total power consumption is determined by the communications among the sensors and with the base station: $P_\text{tot} = P_\text{intra} + P_\text{BS}$. The method in \cite{wsn_datagather} chooses a number of cluster heads, i.e., the sensors aggregating the measurements and transmitting them to the base station, and selects them at random among the sensors. Each sensor which is not a cluster head multiplies its signal coefficient by one or more random Gaussians and sends the value to the closest cluster head. The signal is supposed to be sparse in the DCT domain. For the proposed method, the graph is a nearest-neighbor graph with a radius equal to 0.2. In order to present a fair experiment, we compared the proposed method with \cite{wsn_datagather} in terms of consumed power for the same reconstruction MSE or, equivalently, in terms of reconstruction MSE for a fixed power consumption. This can be observed in Figure \ref{fig:wsn_power_mse} which shows that the proposed method achieves a better power-quality tradeoff, i.e., it is advantageous regardless of the number of clusters for the method in \cite{wsn_datagather}, as it allows to save power for the same quality of reconstruction or, equivalently, to achieve better quality for the same power.

\begin{figure}[t]
\centering
\includegraphics[width=0.55\columnwidth]{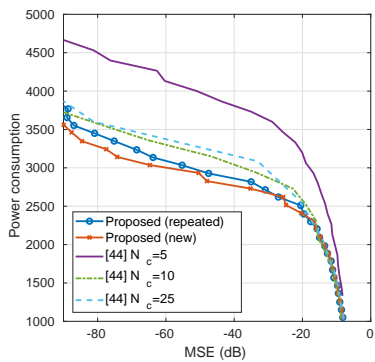}
\caption{Tradeoff between power consumption and reconstruction MSE for data gathering in a wireless sensor network. $N_c$ is the number of cluster heads in \cite{wsn_datagather}. $n=250$, $k=50$.}
\label{fig:wsn_power_mse}
\end{figure}

\section{Conclusions}
\label{sec:conclusions}

We proposed a sampling algorithm for signals defined over graphs that draws concepts from the theories of compressed sensing and graph signal processing. It allows representing a graph signal by a small number of measurements obtained through localized operations, i.e., only involving exchanges of information restricted to local neighborhoods. The method has theoretical guarantees in terms of reconstruction and stability to noise for any graph topology and any orthonormal basis where the signal has a sparse representation. The algorithm is agnostic to information on the signal support and therefore can be used when the signal support is both known or unknown. 

\bibliographystyle{IEEEtran}

\begin{thebibliography}{10}
\providecommand{\url}[1]{#1}
\csname url@samestyle\endcsname
\providecommand{\newblock}{\relax}
\providecommand{\bibinfo}[2]{#2}
\providecommand{\BIBentrySTDinterwordspacing}{\spaceskip=0pt\relax}
\providecommand{\BIBentryALTinterwordstretchfactor}{4}
\providecommand{\BIBentryALTinterwordspacing}{\spaceskip=\fontdimen2\font plus
\BIBentryALTinterwordstretchfactor\fontdimen3\font minus
  \fontdimen4\font\relax}
\providecommand{\BIBforeignlanguage}[2]{{%
\expandafter\ifx\csname l@#1\endcsname\relax
\typeout{** WARNING: IEEEtran.bst: No hyphenation pattern has been}%
\typeout{** loaded for the language `#1'. Using the pattern for}%
\typeout{** the default language instead.}%
\else
\language=\csname l@#1\endcsname
\fi
#2}}
\providecommand{\BIBdecl}{\relax}
\BIBdecl

\bibitem{shuman2013emerging}
D.~I. Shuman, S.~K. Narang, P.~Frossard, A.~Ortega, and P.~Vandergheynst, ``The
  emerging field of signal processing on graphs: Extending high-dimensional
  data analysis to networks and other irregular domains,'' \emph{IEEE Signal
  Processing Magazine}, vol.~30, no.~3, pp. 83--98, May 2013.

\bibitem{sensor_networks}
I.~Jablonski, ``Graph signal processing in applications to sensor networks,
  smart grids and smart cities,'' \emph{IEEE Sensors Journal}, vol.~PP, no.~99,
  pp. 1--1, Dec. 2017.

\bibitem{Lazer721}
D.~Lazer, A.~Pentland, L.~Adamic, S.~Aral, A.-L. Barab{\'a}si, D.~Brewer,
  N.~Christakis, N.~Contractor, J.~Fowler, M.~Gutmann, T.~Jebara, G.~King,
  M.~Macy, D.~Roy, and M.~Van~Alstyne, ``Computational social science,''
  \emph{Science}, vol. 323, no. 5915, pp. 721--723, Feb. 2009.

\bibitem{Davidson1669}
E.~H. Davidson, J.~P. Rast, P.~Oliveri, A.~Ransick, C.~Calestani, C.-H. Yuh,
  T.~Minokawa, G.~Amore, V.~Hinman, C.~Arenas-Mena, O.~Otim, C.~T. Brown, C.~B.
  Livi, P.~Y. Lee, R.~Revilla, A.~G. Rust, Z.~j. Pan, M.~J. Schilstra, P.~J.~C.
  Clarke, M.~I. Arnone, L.~Rowen, R.~A. Cameron, D.~R. McClay, L.~Hood, and
  H.~Bolouri, ``A genomic regulatory network for development,'' \emph{Science},
  vol. 295, no. 5560, pp. 1669--1678, Mar. 2002.

\bibitem{huang2016graph}
W.~Huang, L.~Goldsberry, N.~F. Wymbs, S.~T. Grafton, D.~S. Bassett, and
  A.~Ribeiro, ``Graph frequency analysis of brain signals,'' \emph{IEEE Journal
  of Selected Topics in Signal Processing}, vol.~10, no.~7, pp. 1189--1203,
  Oct. 2016.

\bibitem{thanou2016graph}
D.~Thanou, P.~A. Chou, and P.~Frossard, ``Graph-based compression of dynamic
  3{D} point cloud sequences,'' \emph{IEEE Transactions on Image Processing},
  vol.~25, no.~4, pp. 1765--1778, April 2016.

\bibitem{fracastoro2015predictive}
G.~Fracastoro and E.~Magli, ``Predictive graph construction for image
  compression,'' in \emph{2015 IEEE International Conference on Image
  Processing (ICIP)}, Sept. 2015, pp. 2204--2208.

\bibitem{fracastoro2016graph}
G.~Fracastoro, D.~Thanou, and P.~Frossard, ``Graph transform learning for image
  compression,'' in \emph{2016 Picture Coding Symposium (PCS)}, Dec. 2016, pp.
  1--5.

\bibitem{hu2015multiresolution}
W.~Hu, G.~Cheung, A.~Ortega, and O.~C. Au, ``Multiresolution graph fourier
  transform for compression of piecewise smooth images,'' \emph{IEEE
  Transactions on Image Processing}, vol.~24, no.~1, pp. 419--433, Jan. 2015.

\bibitem{Moura_bigdata}
A.~Sandryhaila and J.~M.~F. Moura, ``Big data analysis with signal processing
  on graphs: Representation and processing of massive data sets with irregular
  structure,'' \emph{IEEE Signal Processing Magazine}, vol.~31, no.~5, pp.
  80--90, Sept. 2014.

\bibitem{Donoho_CS}
D.~L. Donoho, ``Compressed sensing,'' \emph{IEEE Transactions on Information
  Theory}, vol.~52, no.~4, pp. 1289--1306, April 2006.

\bibitem{hammond2011wavelets}
D.~K. Hammond, P.~Vandergheynst, and R.~Gribonval, ``Wavelets on graphs via
  spectral graph theory,'' \emph{Applied and Computational Harmonic Analysis},
  vol.~30, no.~2, pp. 129--150, Mar. 2011.

\bibitem{Leuss_Dual}
G.~Leus, S.~Segarra, A.~Ribeiro, and A.~G. Marques, ``The dual graph shift
  operator: Identifying the support of the frequency domain,'' \emph{arXiv
  preprint arXiv:1705.08987}, 2017.

\bibitem{Tsitsvero_uncertainty}
M.~Tsitsvero, S.~Barbarossa, and P.~D. Lorenzo, ``Signals on graphs:
  Uncertainty principle and sampling,'' \emph{IEEE Transactions on Signal
  Processing}, vol.~64, no.~18, pp. 4845--4860, Sept. 2016.

\bibitem{Puy_random}
G.~Puy, N.~Tremblay, R.~Gribonval, and P.~Vandergheynst, ``Random sampling of
  bandlimited signals on graphs,'' \emph{Applied and Computational Harmonic
  Analysis}, Mar. 2016.

\bibitem{Chen_sampling}
S.~Chen, R.~Varma, A.~Sandryhaila, and J.~Kovacevic, ``Discrete signal
  processing on graphs: Sampling theory,'' \emph{IEEE Transactions on Signal
  Processing}, vol.~63, no.~24, pp. 6510--6523, Dec 2015.

\bibitem{anis2016efficient}
A.~Anis, A.~Gadde, and A.~Ortega, ``Efficient sampling set selection for
  bandlimited graph signals using graph spectral proxies,'' \emph{IEEE
  Transactions on Signal Processing}, vol.~64, no.~14, pp. 3775--3789, July
  2016.

\bibitem{romero2017kernel}
D.~Romero, M.~Ma, and G.~B. Giannakis, ``Kernel-based reconstruction of graph
  signals,'' \emph{IEEE Transactions on Signal Processing}, vol.~65, no.~3, pp.
  764--778, Feb. 2017.

\bibitem{chamon2017greedy}
L.~F. Chamon and A.~Ribeiro, ``Greedy sampling of graph signals,'' \emph{arXiv
  preprint arXiv:1704.01223}, 2017.

\bibitem{Marques_successive}
A.~G. Marques, S.~Segarra, G.~Leus, and A.~Ribeiro, ``Sampling of graph signals
  with successive local aggregations,'' \emph{IEEE Transactions on Signal
  Processing}, vol.~64, no.~7, pp. 1832--1843, April 2016.

\bibitem{wang2016local}
X.~Wang, J.~Chen, and Y.~Gu, ``Local measurement and reconstruction for noisy
  bandlimited graph signals,'' \emph{Signal Processing}, vol. 129, pp.
  119--129, Dec. 2016.

\bibitem{Candes_nearoptimal}
E.~J. Candes and T.~Tao, ``Near-optimal signal recovery from random
  projections: Universal encoding strategies?'' \emph{IEEE Transactions on
  Information Theory}, vol.~52, no.~12, pp. 5406--5425, Dec. 2006.

\bibitem{Varma_spectrum}
R.~Varma, S.~Chen, and J.~Kovacevic, ``Spectrum-blind signal recovery on
  graphs,'' in \emph{2015 IEEE 6th International Workshop on Computational
  Advances in Multi-Sensor Adaptive Processing (CAMSAP)}, Dec. 2015, pp.
  81--84.

\bibitem{sandryhaila2014discrete}
A.~Sandryhaila and J.~M. Moura, ``Discrete signal processing on graphs:
  Frequency analysis.'' \emph{IEEE Transanction on Signal Processing}, vol.~62,
  no.~12, pp. 3042--3054, June 2014.

\bibitem{Candes_decoding}
E.~J. Candes and T.~Tao, ``Decoding by linear programming,'' \emph{IEEE
  Transactions on Information Theory}, vol.~51, no.~12, pp. 4203--4215, Dec.
  2005.

\bibitem{Baraniuk2008}
R.~Baraniuk, M.~Davenport, R.~DeVore, and M.~Wakin, ``A simple proof of the
  restricted isometry property for random matrices,'' \emph{Constructive
  Approximation}, vol.~28, no.~3, pp. 253--263, Dec. 2008.

\bibitem{Mendelson2008}
S.~Mendelson, A.~Pajor, and N.~Tomczak-Jaegermann, ``Uniform uncertainty
  principle for bernoulli and subgaussian ensembles,'' \emph{Constructive
  Approximation}, vol.~28, no.~3, pp. 277--289, Dec. 2008.

\bibitem{Haupt_Toeplitz_IT}
J.~Haupt, W.~Bajwa, G.~Raz, and R.~Nowak, ``Toeplitz compressed sensing
  matrices with applications to sparse channel estimation,'' \emph{IEEE
  Transactions on Information Theory}, vol.~56, no.~11, pp. 5862--5875, Nov.
  2010.

\bibitem{JLCirculant_improved}
J.~Vyb\'iral, ``{A variant of the Johnson-Lindenstrauss lemma for circulant
  matrices},'' \emph{{Journal of Functional Analysis}}, vol.~26, no.~4, pp.
  1096--1105, Feb. 2011.

\bibitem{Yap_block}
H.~L. Yap, A.~Eftekhari, M.~B. Wakin, and C.~J. Rozell, ``The restricted
  isometry property for block diagonal matrices,'' in \emph{2011 45th Annual
  Conference on Information Sciences and Systems}, March 2011, pp. 1--6.

\bibitem{eftekhari2015restricted}
A.~Eftekhari, H.~L. Yap, C.~J. Rozell, and M.~B. Wakin, ``The restricted
  isometry property for random block diagonal matrices,'' \emph{Applied and
  Computational Harmonic Analysis}, vol.~38, no.~1, pp. 1--31, Jan. 2015.

\bibitem{Duarte_structured}
M.~F. Duarte and Y.~C. Eldar, ``Structured compressed sensing: From theory to
  applications,'' \emph{IEEE Transactions on Signal Processing}, vol.~59,
  no.~9, pp. 4053--4085, Sept. 2011.

\bibitem{Valsesia_compressive}
D.~Valsesia and E.~Magli, ``Compressive signal processing with circulant
  sensing matrices,'' in \emph{2014 IEEE International Conference on Acoustics,
  Speech and Signal Processing (ICASSP)}, May 2014, pp. 1015--1019.

\bibitem{sakiyama_2016}
A.~Sakiyama, Y.~Tanaka, T.~Tanaka, and A.~Ortega, ``Efficient sensor position
  selection using graph signal sampling theory,'' in \emph{2016 IEEE
  International Conference on Acoustics, Speech and Signal Processing
  (ICASSP)}, March 2016, pp. 6225--6229.

\bibitem{sakiyama_2017}
------, ``Accelerated sensor position selection using graph localization
  operator,'' in \emph{2017 IEEE International Conference on Acoustics, Speech
  and Signal Processing (ICASSP)}, March 2017, pp. 5890--5894.

\bibitem{Pan_vandermonde}
V.~Y. Pan, ``How bad are {V}andermonde matrices?'' \emph{SIAM Journal on Matrix
  Analysis and Applications}, vol.~37, no.~2, pp. 676--694, 2016.

\bibitem{Garey_computers}
M.~R. Garey and D.~S. Johnson, \emph{{Computers and Intractability; A Guide to
  the Theory of NP-Completeness}}.\hskip 1em plus 0.5em minus 0.4em\relax New
  York, NY, USA: W. H. Freeman \& Co., 1990.

\bibitem{Jia2002}
L.~Jia, R.~Rajaraman, and T.~Suel, ``An efficient distributed algorithm for
  constructing small dominating sets,'' \emph{Distributed Computing}, vol.~15,
  no.~4, pp. 193--205, Dec. 2002.

\bibitem{Krause2008}
A.~Krause, A.~Singh, and C.~Guestrin, ``Near-optimal sensor placements in
  gaussian processes: Theory, efficient algorithms and empirical studies,''
  \emph{J. Mach. Learn. Res.}, vol.~9, pp. 235--284, Jun. 2008.

\bibitem{krahmer2014suprema}
F.~Krahmer, S.~Mendelson, and H.~Rauhut, ``Suprema of chaos processes and the
  restricted isometry property,'' \emph{Communications on Pure and Applied
  Mathematics}, vol.~67, no.~11, pp. 1877--1904, Nov. 2014.

\bibitem{Gut_probability}
A.~Gut, \emph{{An Intermediate Course in Probability}}.\hskip 1em plus 0.5em
  minus 0.4em\relax Springer, 2009.

\bibitem{hogg2014probability}
R.~V. Hogg, E.~Tanis, and D.~Zimmerman, \emph{Probability and statistical
  inference}.\hskip 1em plus 0.5em minus 0.4em\relax Pearson Higher Ed, 2014.

\bibitem{perraudin2014gspbox}
N.~Perraudin, J.~Paratte, D.~Shuman, L.~Martin, V.~Kalofolias,
  P.~Vandergheynst, and D.~K. Hammond, ``{GSPBOX}: A toolbox for signal
  processing on graphs,'' \emph{arXiv preprint arXiv:1408.5781}, 2014.

\bibitem{wsn_datagather}
M.~T. Nguyen and K.~A. Teague, ``Compressive sensing based data gathering in
  clustered wireless sensor networks,'' in \emph{2014 IEEE International
  Conference on Distributed Computing in Sensor Systems}, May 2014, pp.
  187--192.

\end{thebibliography}

\end{document}